\documentclass[journal,12pt,draftclsnofoot,onecolumn]{IEEEtran}

\usepackage{cite}
\ifCLASSINFOpdf
\fi
\usepackage{amsmath,amsfonts,amsthm} % Math packages
\usepackage{bbm}
\usepackage{float}
\usepackage{graphicx}
\usepackage{algorithm}
\usepackage{algpseudocode}
\usepackage{subcaption}
\usepackage{siunitx}
\usepackage{tikz}
\usepackage{pgfplots}
\usepackage{epstopdf}
\usepackage{multirow}
\usepackage[font=scriptsize,skip=0pt]{caption}
\usepackage{soul}
\usepackage{todonotes}
\usepackage{color, colortbl}
\usepackage{array}
\usetikzlibrary{patterns}
\usepackage{balance}

\usepackage{enumitem}
\usepackage{footmisc}
\usepackage{hyperref}
\usetikzlibrary{shapes.geometric}
\newcommand{\comment}[1]{}
\newcommand{\argmax}{\arg\!\max}

\graphicspath{{images/}}
\newtheorem{theorem}{Theorem}

\newtheorem{lemma}[theorem]{Lemma}

\definecolor{Gray}{gray}{0.9}
\definecolor{LightCyan}{rgb}{0.88,1,1}

\algnewcommand{\IIf}[1]{\State\algorithmicif\ #1\ \algorithmicthen}
\algnewcommand{\IElse}[2]{\State\algorithmicelse\ #2\ }
\algnewcommand{\EndIIf}{\unskip\ \algorithmicend\ \algorithmicif}

\begin{document}
\setlength{\parskip}{0em}

\title{Layered Coding for Energy Harvesting Communication Without CSIT}
	\author{Rajshekhar~Vishweshwar~Bhat,~\IEEEmembership{Graduate~Student~Member,~IEEE,}
		Mehul~Motani,~\IEEEmembership{Senior~Member,~IEEE,}
		and~Teng~Joon~Lim,~\IEEEmembership{Fellow,~IEEE}}
	\maketitle
\vspace{ -1.2cm}
\begin{abstract}
	\vspace{ -.2cm}
Due to stringent constraints on resources, it may be infeasible to acquire the current channel state information at the transmitter in energy harvesting communication systems. 
In this paper, we optimize an energy harvesting transmitter, communicating over a slow fading channel, using layered coding. The transmitter has access to the channel statistics, but does not know the exact channel state. In layered coding, the codewords are first designed for each of the channel states at  different rates, and then the codewords are either 
time-multiplexed or superimposed before the transmission, leading to two transmission strategies. 
The receiver then decodes the information adaptively based on the realized channel state. The transmitter is equipped with a finite-capacity battery having non-zero internal resistance. In each of the transmission strategies, we first formulate and study an average rate maximization problem with non-causal knowledge of the harvested power variations. We also highlight the structural properties of the optimal solution. Further, assuming statistical knowledge and causal information of the harvested power variations, we propose a sub-optimal algorithm, and compare with the stochastic dynamic programming based solution and a greedy policy. By numerical simulations, we also show that the internal resistance significantly affects the system performance. 
\end{abstract}

\IEEEpeerreviewmaketitle
\vspace{ -.2cm}
\section{Introduction}
\vspace{ -.1cm}
Recently, there has been a tremendous interest in sustainable wireless communication systems powered solely from  natural energy harvesting (EH) sources \cite{IR, survey1,survey2,survey3,Vaneet,ulukusTGCN,vincent}. Though extremely promising, it poses several challenges in system design as the power generated from EH sources randomly varies with time. The harvested energy needs to be stored in and drawn from the batteries, at appropriate rates, for reliable system operation. In the process, due to the source and load power fluctuations, the charge and discharge powers of the battery are more variable and unpredictable than in conventional systems \cite{IR}. This necessitates a fundamental change in the  way we store and use the harvested energy mainly given that the battery charge/discharge efficiencies depend on the charge and discharge powers \cite{Krieger, IR} when the batteries exhibit non-negligible internal resistances. 

In this paper, we consider an EH transmitter communicating over a \emph{quasi-static channel} -- a slow fading channel in which the fading realizations remain constant for a certain period of time, known as the coherence block, and change independently across the blocks. 
In EH systems, it may be infeasible to acquire the current channel state information (CSI) at the transmitter (CSIT)  due to stringent constraints on the resources, such as energy, bandwidth and processing capabilities  \cite{ISIT14SFDCSIT}. Hence, we assume that the transmitter only knows the channel distribution and the current CSIT is unavailable. Further, the receiver has  perfect CSI. 

In many emerging applications, such as Internet of Things and Machine-type communications, practical delay and latency requirements may prohibit a codeword from spanning multiple coherence blocks  \cite{Codelength4G}. Hence, the codewords do not experience the average fading process. However, the codeword lengths can be long enough to achieve the reliable communication using channel codes. In such cases, it has been shown that layered coding, a technique which facilitates the adaptation of the transmission rate to the realized channel state, achieves a higher throughput than transmission with a single fixed rate \cite{BCoriginal, Diggavi,Berry,Goldsmith,Erkip05,Shamai}. In this work, our objective is to maximize the average achievable rate using layered coding by optimally managing the battery charging and discharging schedules across $K$ coherence blocks (frames). 

While there are many works in EH communication which assume perfect CSIT with causal and non-causal knowledge of the harvested power \cite{IR, survey1,survey2,survey3,Vaneet}, works considering imperfect, delayed or absent CSIT have been relatively scarce. \cite{Training} optimizes the resource allocation for a training-based channel estimation. Outage minimization problems in EH nodes communicating over slow fading channels have been  considered in \cite{outage,outage1,outage2,outage3,IK}. \cite{imperfect_icc} presents optimal transmission policies with imperfect CSIT and without CSIT in fast fading channels. In non-EH communication, rate maximization and distortion minimization problems in quasi-static channels without CSIT have been  studied \cite{BCoriginal, Diggavi,Berry,Goldsmith,Erkip05,Shamai}.  EH-powered broadcasting nodes transmitting over static channels have been considered in \cite{BC-cutoff,BC-cutoff-finite}. 
The major challenge in this work is in  accounting for the circuit cost, internal resistance and capacity limitation of the battery, and EH-related constraints in the rate maximization problem using layered coding.
The main contributions of this paper are as follows:
\begin{itemize}[leftmargin=*]
	\item  We formulate and analyze average rate maximization problems in the \emph{offline} case with non-causal knowledge of the harvested power variations under two transmission strategies  wherein the codewords are either time-multiplexed or superimposed before the transmission.   
	\item For the superposition coding based strategy, we provide a simple and concise interpretation, referred to as layered water-filling algorithm, for the optimal solution in an ideal single frame case, based on which we present efficient algorithms in more general cases.     
	\item With statistical knowledge and causal information of harvested power, we propose a sub-optimal \emph{online} algorithm based on the offline solution and compare with stochastic dynamic programming based solution and a greedy policy under the time-multiplexed and superposition coding based transmission strategies. 
\end{itemize} 
\begin{figure*}[t]
	\centering
	\begin{tikzpicture} [scale=2.5]
	\draw  (1,0.125) -- (1.28,0.125) node [midway,above] {\scriptsize $U$}  node [left,xshift=-.85cm] {\scriptsize harvested power};
	\draw [->] (1.1,0.125) -- (1.15,0.125);
	
	\draw  (1.35,0.125) circle (3pt) node [align=right]  at (1.82,0.12)  {\scriptsize Power Splitter};
	\draw  (1.35,0.225) -- (1.35,0.44) node[right,midway] {\scriptsize $\alpha(t) U$};
	\draw  (1.35,0.015) -- (1.35,-0.19) node[right,midway,yshift=.2cm] {\scriptsize $(1-\alpha(t))U$};
	\draw [->] (1.35,0.225) -- (1.35,0.35);
	\draw [->] (1.35,0) -- (1.35,-0.1);
	\draw (1.28,0.125) -- (1.35,0.225);
	\draw (1.28,0.125) -- (1.35,0.015);
	
	\draw (1.35,0.44) -- (3.2,.44) node[above,midway] {\scriptsize direct path with zero losses};
	\draw [->] (1.35,0.44) -- (2.4,0.44);
	\draw (1.35,-0.19) -- (1.75,-0.19);
	\draw (1.75,-0.30) rectangle (2.75,-0.08) node [align=center]  at (2.25,-0.19)  {\scriptsize Battery (r, $B_{\mathrm{max}}$)};
	\draw (2.75,-0.19) -- (3.2,-0.19) node [midway,above] {\scriptsize $d(t)$};
	\draw [->](2.85,-0.19) -- (2.9,-0.19);
	\draw  (3.2,0.13) circle (3pt) node {\small +} node at (2.65,0.125) {\scriptsize Power Combiner};
	\draw (3.2,-0.19) -- (3.2,.03);
	\draw [->] (3.2,-0.1) -- (3.2,-0.08);
	\draw (3.2,0.44) -- (3.2,.23);	
	\draw [->] (3.2,0.35) -- (3.2,.33);
	\draw (3.3,0.13) -- (3.4,0.13);
	\draw [->](3.3,0.13) -- (3.36,0.13);
	\draw (3.4,-0.25) rectangle (4.25,.5) node [align=center] at (3.85, 0.25)  {\scriptsize Transmitter} node  [align=center] at (3.842,0.02)  {\scriptsize circuit cost $P_C$};
	\draw (4.25,.13) -- (4.5,0.13);
	\draw [->] (4.25,.13) -- (4.35,0.13);
	\tikzset{
		buffer/.style={
			draw,
			shape border rotate=-90,
			isosceles triangle,
			isosceles triangle apex angle=60,
			node distance=.03cm,
			minimum height=.03em
		}
	}
	\draw (4.5,0.13) -- (4.5,.32);
	\node at (4.5,.46)[buffer]{};
	
	\draw (4.6,.34) -- (5.15,.34) node [above,midway] {\scriptsize $H$};
	\node at (5.25,.46)[buffer]{};	
	\draw (5.25,0.13) -- (5.25,.32);
	\draw (5.25,.13) -- (5.5,0.13);	
	\draw [->](5.25,.13) -- (5.35,0.13);		
	\draw (5.5,-0.25) rectangle (6.25,.5) node [align=center] at (5.85, 0.12)  {\scriptsize Receiver};
	\end{tikzpicture}
	\caption{\scriptsize The dual-path EH communication system. In a frame of length $\tau$ seconds, at any time $t$ ($0\leq t\leq \tau$), instantaneous fraction, $\alpha(t)$ ($0\leq\alpha(t)\leq 1$) of the harvested power ($U$ \si{\watt}) can be directed to the load through the direct path. The remaining power is directed to the battery having internal resistance of $r$ \si{\ohm} and capacity $B_{\mathrm{max}}$ \si{\joule}. The battery is discharged at an instantaneous rate $d(t)$ W.  The transmitter consumes $P_C$ \si{\watt} for its operation  during transmission but does not consume any power when not transmitting. }
	\label{fig:BD} 
\end{figure*}
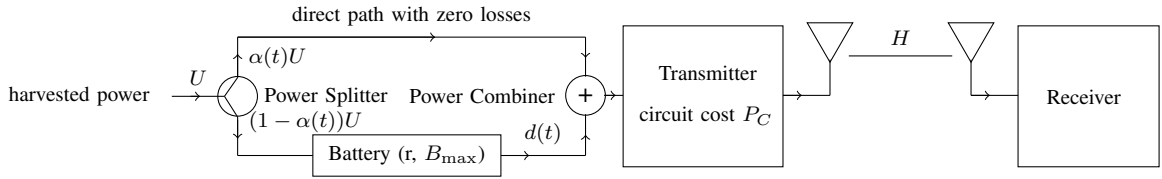
The remainder of the paper is organized as follows. The system model and assumptions are presented in Section \ref{sec:model}. In Section \ref{sec:problem formulation}, we formulate the generic optimization problem which is solved in Section \ref{sec:LTM_strategy} and Section \ref{sec:LSC_strategy}. Online policies are presented in Section \ref{statistical}. Numerical results are presented in Section \ref{sec:numerical results} followed by concluding remarks in Section \ref{sec:reflections}.

\section{System Model and Assumptions}\label{sec:model}
\subsection{Block Diagram and System Operation}
The block diagram of the system under study is given in Fig.  \ref{fig:BD}. The power splitter divides the harvested power, $U$ \si{\watt}, to simultaneously charge the battery and power the transmitter directly.  The power combiner combines the power drawn from the battery and the direct path. The transmitter consumes $P_C$ W for circuit operation  during transmission but does not consume any power when not transmitting \cite{IR,ulukusTGCN}. 
%We assume the battery cannot be charged and discharged simultaneously. This assumption is both practical and without loss of generality \cite{IR}.
We assume that the internal resistance and the maximum capacity of the battery are $r$ \si{\ohm} and $B_{\mathrm{max}}$ \si{\joule}, respectively. 

\subsection{Battery Charge and Discharge Model}
As in \cite{IR}, we model the battery as a voltage source/sink with a series internal resistance of $r$ \si{\ohm}. %We note that the model is a well accepted model in circuit theory. 
In practice, when the EH source (transmitter) attempts to charge (discharge) the battery, a fraction of the charging (discharging) power is lost in the form of heat dissipated by the internal resistance of the battery. To describe this impact of the internal resistance, we present a block diagram in Fig. \ref{fig:flow-diagram} where the battery with internal resistance is depicted as a ideal battery with two additional blocks that model the effect of the internal resistance. When the power is driven to the battery at $V$ \si{\watt},  the rate at which energy accumulates in the battery is  $\mathcal{F}_c(V,r)$ \si{\watt}. The remaining $V-\mathcal{F}_c(V,r)$ \si{\watt} is lost in the internal resistance. Similarly, when the battery is discharged at $d$ \si{\watt},  the rate at which energy is available at the load is $\mathcal{F}_d(d,r)$ \si{\watt} and the remaining  $d-\mathcal{F}_d(d,r)$ \si{\watt} is lost in the internal resistance.
Based on  \cite{IR}, we assume that the functions $\mathcal{F}_c(V,r)$ and $\mathcal{F}_d(d,r)$ have the following properties. 
\begin{itemize}[leftmargin=*]
	\item $\mathcal{F}_c(V,r)$ and $\mathcal{F}_d(d,r)$ are concave functions of $V$ and $d$, respectively, for a fixed internal resistance $r$.  
	\item $\mathcal{F}_c(V,r)\leq V$ and $\mathcal{F}_d(d,r)\leq d$ for a fixed $r$. 
	\item   $\mathcal{F}_c(V,r)$ and $\mathcal{F}_d(d,r)$ are decreasing functions of $r$ for  fixed values of $V$ and $d$. 
\end{itemize}
In this work, the internal resistance, $r$ is not an optimization variable. Hence, in the rest of the paper, we denote $\mathcal{F}_c(V,r)$ as $\mathcal{F}_c(V)$, and $\mathcal{F}_d(d,r)$ as $\mathcal{F}_d(d)$ for brevity. 
In this work, our analysis is fully general in the sense that it holds for any $\mathcal{F}_c(V)$ and $\mathcal{F}_d(d)$ with the above properties. 

\begin{figure}[t]
	\centering
	\begin{tikzpicture}[scale=1.5]
	\draw [dotted,fill=black!5](1.3,-0.5) rectangle (7.2,0.5) node [above, midway,yshift=0.6cm] {\scriptsize Battery with Internal Resistance ($r$ \si{\ohm})};   
	\draw [->](0,0) -- (1.5,0) node [midway,above,xshift=-0.3cm] {\scriptsize $V=(1-\alpha)U$};
	\draw (1.5,-0.3) rectangle (2.5,0.3) node [pos=0.5] {\scriptsize $\mathcal{F}_c(\cdot,r)$};
	\draw [->] (2.5,0) -- (3.75,0) node [midway,above] {\scriptsize $\mathcal{F}_c(V,r)$};
	\draw (3.75,-0.35) rectangle (5.25,0.35) node [pos=0.5,yshift=0.13cm] {\scriptsize Ideal }; 
	\draw (3.75,-0.35) rectangle (5.25,0.35) node [pos=0.5,yshift=-0.13cm] {\scriptsize Battery}; 
	\draw [->] (5.25,0) -- (6,0) node [midway,above] {\scriptsize $d$};
	\draw (6,-0.3) rectangle (7,0.3) node [pos=0.5] {\scriptsize $\mathcal{F}_d(\cdot,r)$};
	\draw [->] (7,0) -- (8.25,0) node [midway,above,xshift=0.23cm] {\scriptsize $\mathcal{F}_d(d,r)$}; 
	\end{tikzpicture}
	\caption{The battery with internal resistance is depicted as the ideal battery with two additional blocks that model the effect of the internal resistance. When the power is driven to the battery at $V$ \si{\watt}, the rate at which energy accumulates in the battery is  $\mathcal{F}_c(V,r)$ \si{\watt}. Similarly, when the battery is discharged at $d$ \si{\watt},  the rate at which energy is available at the load is $\mathcal{F}_d(d,r)$ \si{\watt}.}
	\label{fig:flow-diagram}
\end{figure}
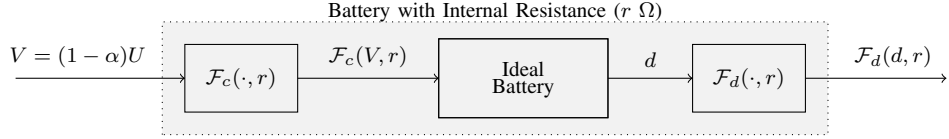

%It has been shown in \cite{IR} that the internal charging rate (the rate at which energy is stored in the battery after losses) of the battery, represented as $\mathcal{F}_c(V)$, is a concave function of the external charging rate (when measured at the  external leads), $V$ and the external discharge rate of the battery, $\mathcal{F}_d(d)$ is a concave function of the internal discharge rate $d$. Also, $\mathcal{F}_c(V)\leq V$ and $\mathcal{F}_d(d)\leq d$ for a fixed $r$. 
\subsection{Frame Structure}
We assume that the transmission frame length, denoted by $\tau$ (seconds), is smaller than the coherence block length. When the total available energy in a frame is lower than the total energy required to operate the system,  the transmission can occur only over a fraction of the frame duration \cite{IR}. Let $\phi \; (\phi\leq \tau)$ be the duration for which the system is transmitting in a frame. The power splitting ratio, the fraction of the harvested power directly used for the transmission, is $\alpha(t)$ at time $t$, where $0 \leq \alpha(t) \leq 1$. The frame structure (see Fig. \ref{fig:Frame}) is as follows. 
\begin{itemize}[leftmargin=*]
	\item \emph{Non-transmission} phase: over the time duration $[0,\tau-\phi)$, the battery is charged at the optimal uniform rate, $V_a^*=(1-\alpha_a^*)U$ W, where $\alpha_{a}^*=\max_{0\leq \alpha\leq 1} \mathcal{F}_c\left((1-\alpha)U\right)$. No information is transmitted in this phase. 
	\item \emph{Transmission} phase: over the time duration $[\tau-\phi,\tau]$, information is transmitted while the battery is being charged at the instantaneous rate, $V_b(t)=(1 - \alpha_b(t))U$ \si{\watt} and discharged at the instantaneous  rate $d(t)$ W. Whether the battery is being charged or discharged,  $\alpha_b(t)$ fraction  of the harvested power is directly delivered to the transmitter.   
\end{itemize}
We note that the frame structure in Fig. \ref{fig:Frame} has been shown to be necessary and sufficient to extract the maximum possible performance from the system \cite{IR}. 

 \begin{figure*}[t]
	\centering
	\begin{tikzpicture} [scale=10]
	\draw (0,0) -- (.4,0) node [below, midway] {\scriptsize $\alpha(t)=\alpha_a^*$}	;
	\draw (0,0) -- (.4,0); 	
	\draw (.4,0) -- (1,0) node [below, midway] {\scriptsize $\alpha(t)=\alpha_b(t)$}	;
	\draw (.4,0) -- (1,0) node [below, midway, yshift=-0.35cm] {\scriptsize $d(t)$};		
	
	\draw (0,0.05) -- (0, -0.05) node [below] {\scriptsize $0$};	;
	\draw (1,0.05) -- (1, -0.05) node [below] {\scriptsize $ \tau$};	
	\draw (.4,0.05) -- (.4, -0.05) node [below] {\scriptsize $\tau-\phi$};	
	
	\draw (0,0) -- (.4,0) node [above, midway] {\scriptsize non-transmission phase};
	\draw (.4,0) -- (1,0) node [above, midway] {\scriptsize transmission phase};	
	\end{tikzpicture}
	\caption{\scriptsize The communication frame structure adopted in the paper. The frame length is $\tau$ seconds.  During $[0,\tau-\phi)$, the battery is charged at $V_a^*=(1-\alpha_a^*)U$ \si{\watt} and the discharge power is zero. In this period, no information is transmitted.   During $[ \tau-\phi,\tau]$, information is transmitted  and the battery is charged at an instantaneous rate $V_b(t)=(1-\alpha_b(t))U$ \si{\watt} and discharged at  an instantaneous rate $d(t)$ \si{\watt}. }
	\label{fig:Frame}
%	\vspace{-1cm}
\end{figure*}
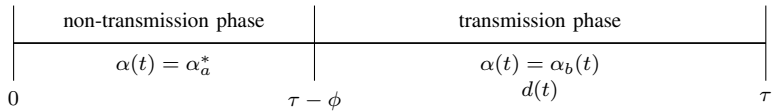

\subsection{Channel and Energy Models}
The communication is over a quasi-static channel with a random channel power gain $H$, corrupted by an additive white Gaussian noise  having  variance $N_0$ \si{\watt/\hertz}. We assume that  $H$ has $N$ non-zero discrete realizations as follows. For $i=1,\ldots,N$, the channel gain realization is $h_i$ with probability $p_i$, where $h_1<h_2<\ldots<h_N$, $0\leq p_i \leq 1$ and $\Sigma_{i=1}^Np_i=1$. The harvested power in any frame $k$ is a random variable $U_k$ whose realization remains constant in a frame and changes independently across frames, and $U_1,\ldots,U_K$, are independent and identically distributed. The value of $U_k$ is measured at the start of frame $k$, and therefore can be used in system optimization in frame $k$.
\subsection{Rate Function}
We assume that the maximum achievable data rate is $\mathcal{G}(x)$ when the instantaneous signal-to-noise ratio  at the receiver is $x$, and that $\mathcal{G}(x)$ is a concave, strictly increasing, invertible function of $x$. Most practical coded-modulation schemes exhibit such a relationship \cite{Goldsmithbook}. If the transmission rate is greater than $\mathcal{G}(x)$, the information cannot be decoded and an outage event is declared at the receiver.   Without loss of generality, we assume  the unit noise power spectral density, i.e., $N_0=1$, for the analysis. 

\section{Problem Formulation}\label{sec:problem formulation}
In this work, we adopt the following layered coding technique. The codewords (layers) are first designed for each of the channel states at different rates. The layers corresponding to larger (smaller) values of $h_i$'s are referred to as  higher (lower) layers. The rates of the layers are designed such that when the channel gain realization is $h_i$, layers $1$ to $i$ can be successfully decoded. To transmit all the layers  in the given frame, the layers are either time-multiplexed or superimposed. This leads to two transmission strategies -- the layered time-multiplexing (LTM) strategy, where the layers are time-multiplexed and the layered superposition coding (LSC) strategy,  where the layers are superimposed before the transmission. The receiver then decodes the information adaptively based on the realized channel state. 

In any given transmission strategy, let the instantaneous power allocated to layer $i$ of frame $k$ be $P_{i,k}(t)$. Let $R_{k}\left(P_{1,k}(t),\ldots,P_{N,k}(t)\right)$, $\alpha_{b_k}(t)$, $d_k(t)$  and $\phi_k$ denote the average rate, instantaneous power splitting ratio, instantaneous discharge power and transmission duration in any frame $k$, respectively.  In the offline case, when the values of $U_1,\ldots,U_K$ are known at the start of the first frame, the optimization problem of interest is: 
\begin{subequations}\label{eq:opt_gen}
	\begin{alignat}{2}
	\mathrm{P_{\mathrm{gen}}}:&\underset{\substack{ \{P_{1,k}(t),\ldots,P_{N,k}(t),d_k(t)\}\\\{0\leq \phi_k\leq \tau,\alpha_{b_k}(t)\}}}{\text{maximize}}\;\;\frac{1}{K}\sum_{k=1}^{K} \phi_k  R_{k}\left(P_{1,k}(t),\ldots,P_{N,k}(t)\right) \quad \text{s.t.}\\
	&  \sum_{k=1}^{j}\int_{0}^{\tau'}\left(\sum_{i=1}^{N}P_{i,k}(t)-g_k(\alpha_{b_{k}}(t),d_{k}(t))\right)dt\leq 0,  \;\; \text{for} \;\; 0\leq \tau' \leq  \tau  \label{eq:0c1} \\
	& P_{1,k}(t),\ldots,P_{N,k}(t), d_k(t) \geq 0, 0\leq \phi_k\leq \tau, 0\leq \alpha_{b_k}\leq 1, \;\; \text{for} \;\; 0\leq t\leq \tau
	\end{alignat}
\end{subequations}
for $j,k=1, \ldots,K$,  where \eqref{eq:0c1} is the energy causality constraint,  $g_k(\alpha_{b_{k}}(t),d_{k}(t))$ is  the instantaneous amount of energy available in frame $k$ including losses. 

In the next two sections, we reformulate and solve $\mathrm{P_{\mathrm{gen}}}$ in \eqref{eq:opt_gen} in LTM and LSC strategies. Since $K$ is a constant, without loss of generality, we maximize the sum rate across the $K$ frames instead of the average rate. In both the strategies, we note that the reformulated optimization problems are non-convex and transform them into equivalent problems which are non-convex in general, but, convex when $B_{\mathrm{max}}=\infty$. Based on the analytical solutions to the equivalent problems for $B_{\mathrm{max}}=\infty$, we solve the equivalent problems for arbitrary $B_{\mathrm{max}}$ for $K=1,2$. These optimal solutions are used to propose sub-optimal online algorithms in each of the strategies.

It is well known that the LSC strategy performs better than the LTM strategy \cite{Erkip05}. However, the implementation of the LSC strategy is complex as the power allocation across the layers are coupled, symbols are superimposed and decoding is sequential \cite{Erkip05}. In both the strategies,  the amount of information reliably decoded by the receiver depends on the realized channel state. At the end of every frame, the receiver sends an acknowledgment to the transmitter to indicate the amount of information decoded in the current frame. We note that the overheads associated with communicating an acknowledgment is negligible compared to the data payload.  Hence, in many systems, though acquiring the current CSIT is impractical, communicating the acknowledgment in every frame is feasible.

Before we proceed, we present an important result in the following lemma. This result will be used in the proofs of later results. 
\begin{lemma}\label{lemma:c-d}
	When $\mathcal{F}_c(\cdot)$ and  $\mathcal{F}_d(\cdot)$ are strictly concave functions, it is not optimal to charge and discharge the battery simultaneously, i.e.,  $(1-\alpha_{b_k}(t)^*)d_k(t)^*=0$. 
\end{lemma}
\begin{proof}
%	See the proof of Lemma 1 in \cite{IR_NO_CSIT}. 
	Let $\{\alpha_{b_i}(t),d_{b_i}(t)\}_{i=1}^K$, be a feasible solution that satisfies $(1-\alpha_{b_i}(t))d_{b_i}(t)>0$ for some $k$. 
	In this case, the total power available at the load in frame $k$, $P_k(t)=\alpha_{b_k}(t)U_k+\mathcal{F}_d(d_k(t))$. While keeping the net charge rate (the difference between the rate at which energy accumulates in the battery and the rate at which energy gets depleted from the battery) the same, whenever $d_k(t)\geq \mathcal{F}_c((1-\alpha_{b_k}(t))U_k)$, we can always discharge the battery at $d_k(t)'=d_k(t)-\mathcal{F}_c((1-\alpha_{b_k}(t))U_k)$ and $\alpha_{b_k}(t)'=1$ obtaining the transmit power, $P_{k}(t)'=U_k+\mathcal{F}_d(d_k(t)')> P_k(t)$. Similarly, when  $d_k(t)<\mathcal{F}_c((1-\alpha_{b_k}(t))U_k)$, we can always charge the battery with $\alpha_{b_k}(t)'$ such that $\mathcal{F}_c((1-\alpha_{b_k}(t))U_k)-d_k(t)=\mathcal{F}_c((1-\alpha_{b_k}(t)')U_k)$ and $d_k(t)'=0$ obtaining  $P_{k}(t)'=\alpha_{b_k}(t)'U_k> P_k(t)$. Hence, we can always replace any $(\alpha_{b_k}(t),d_k(t)), \;(1-\alpha_{b_k}(t))d_k(t)=0 $ with  $(\alpha_{b_k}(t)',d_k(t)'),\;(1-\alpha_{b_k}(t)')d_k(t)'=0$ and get a higher total power at the load for the same net charge rate. The inequality, $P_k(t)'>P_k(t)$ follows from the strict concavity of $\mathcal{F}_c(\cdot)$ and  $\mathcal{F}_d(\cdot)$ and the fact that    $\mathcal{F}_c(V)\leq V$ and $\mathcal{F}_d(d)\leq d$ for a fixed $r$.   
\end{proof}
We note that when the internal resistance, $r$ is non-zero, $\mathcal{F}_c(\cdot)$ and  $\mathcal{F}_d(\cdot)$ are strictly concave functions \cite{IR}. 

Now, we consider the LTM strategy. 
\section{LTM Strategy}\label{sec:LTM_strategy}
In the LTM strategy, the frame is divided into $N$ disjoint partitions. The length of partition $i$ of any frame $k$ is $l_{i,k}$ and  $\sum_{i=1}^{N}l_{i,k}=\tau$. The message is coded in $N$ layers and layer $i$ is transmitted in partition $i$  with constant power $P_{i,k}$, i.e., $P_{i,k}$ is not a function of time, $t$, in any given partition\footnote{Due to concavity of $\mathcal{G}(x)$, it can be shown that transmission with an constant power is optimal for the given total energy and time constraints.}. Hence, for any frame $k$,  $P_{i,k}=\alpha_{b_k}(t)U_k+\mathcal{F}_d(d_k(t))-P_C, \phi_k+\Sigma_{j=1}^{i-1}l_{j,k} \leq t \leq \phi_k+\Sigma_{j=1}^{i}l_{j,k}$, where $\Sigma_{j=1}^{0} l_{j,k}=0$. Further, due to concavity of charging and discharging functions, $\mathcal{F}_c(.)$ and $\mathcal{F}_d(.)$, charging and discharging at constant rates is optimal \cite{IR}.  Hence, we assume that $\alpha_{b_k}(t)=\alpha_{b_{i,k}}$ and $d_k(t)=d_{i,k}$ over the $i$th partition in frame $k$, i.e., for $t\in \left[\phi_k+\Sigma_{j=1}^{i-1}l_{j,k},\phi_k+ \Sigma_{j=1}^{i}l_{j,k}\right] $.  Consequently,  $P_{i,k}=\alpha_{b_{i,k}}U_k+\mathcal{F}_d(d_{i,k})-P_C$.  
% CITE Optimal packet scheduling in an energy harvesting communication system, 

In any frame $k$, the rate of layer $i$ is designed to be equal to the channel capacity of a static channel with gain $h_i$, i.e., the transmission rate in partition  $i$  of frame $k$ is  $R_{i,k}=\mathcal{G}(h_iP_{i,k})$.  When the actual channel realization is $h_j$, the channel capacity in layer $i$ is $F_{i,k}^{(j)}=\mathcal{G}(h_jP_{i,k})$.  Since $h_1<h_2<\ldots<h_N$, we note that $F_{m,k}^{(j)} \geq R_{m,k},\;\forall\;m\in\{1,\ldots,j\} $ and $F_{m,k}^{(j)} < R_{i,k},\;\forall\;m\in\{j+1,\ldots,N\} $.  Hence, we can successfully decode only layers up to and including layer $j$ and, the higher layers will be in outage.  Consequently, the number of bits successfully transmitted in the frame when $H=h_j$ is $\Sigma_{i=1}^j l_{i,k} R_{i,k}$ and the average rate in frame $k$ is given by
\begin{align}\label{eq:avgRate_LTM}
R_{k}^{(\mathrm{LTM})}=\sum_{j=1}^{N}p_j\sum_{i=1}^{j} l_{i,k} R_{i,k}=\sum_{i=1}^{N}l_{i,k}R_{i,k}\sum_{j=i}^{N}p_j=\sum_{i=1}^{N}q_il_{i,k}R_{i,k}=\sum_{i=1}^{N}q_il_{i,k}\mathcal{G}(h_iP_{i,k})
\end{align}
where $q_i=\Sigma_{j=i}^{N}p_j$. 
In frame $k$, the amount of energy stored in the battery in the non-transmission phase is $(\tau-\phi_k) \mathcal{F}_c(V_{a_k}^*)$.  The amount of energy stored in and drawn from the battery in partition $i$ are $l_{i,k}\mathcal{F}_c(V_{i,k})$ and $l_{i,k}d_{i,k}$, respectively, where $V_{i,k}=(1-\alpha_{b_{i,k}})U_k$. Define $E_{i,k}=l_{i,k}d_{i,k}-l_{i,k}\mathcal{F}_c(V_{i,k})$. We now describe the EH-related constraints. For simplicity, we start with the first frame. Since the amount of energy drawn from the battery cannot be greater than the amount of energy stored in the battery  (the energy causality constraint), we have, $\sum_{i=1}^{m}l_{i,1}d_{i,1}\leq \sum_{i=1}^{m}l_{i,1}\mathcal{F}_c(V_{i,1})+(\tau-\phi_1) \mathcal{F}_c(V_{a_1}^*)+B_0$,  or equivalently, $\sum_{i=1}^{m}E_{i,1}-(\tau-\phi_1) \mathcal{F}_c(V_{a_1}^*)-B_0\leq 0$,  for $m=1,\ldots,N$, where $B_0$ is the initial energy stored in the battery. 
Similarly, in order to avoid  energy overflow in the battery, the amount of energy stored in the battery at any time must be less than or equal to the battery capacity (the battery capacity constraint), i.e., $(\tau-\phi_1) \mathcal{F}_c(V_{a_1}^*)+B_0-\sum_{i=1}^{m}E_{i,1}\leq B_{\mathrm{max}}$,  for $m=1,\ldots,N$.  In general, the energy causality and  battery capacity constraints in partition $m$ of frame $k$ are respectively given  by, 
\begin{align}
&\sum_{i=1}^{m}E_{i,k} +\sum_{j=1}^{k-1}\sum_{i=1}^{N} E_{i,j}-\sum_{j=1}^{k}(\tau-\phi_j) \mathcal{F}_c(V_{a_j}^*)-B_0\leq 0 \label{eq:p0_ltm_c1}\\
&\sum_{j=1}^{k}(\tau-\phi_j) \mathcal{F}_c(V_{a_j}^*)+B_0-\sum_{i=1}^{m}E_{i,k} -\sum_{j=1}^{k-1}\sum_{i=1}^{N} E_{i,j}-B_{\mathrm{max}}\leq 0  \label{eq:p0_ltm_c2}
\end{align} 
Hence, to maximize the sum rate over $K$ frames, $\mathrm{P_{\mathrm{gen}}}$ in \eqref{eq:opt_gen} can be reformulated as, 
\begin{subequations}\label{eq:p0_ltm}
	\begin{alignat}{2}
	\mathrm{P_{0-LTM}}: &\underset{\{\alpha_{b_{i,k}}, {l_{i,k}},d_{i,k},\phi_k\}}{\text{maximize}}\;\sum_{k=1}^{K}\sum_{i=1}^{N}q_il_{i,k} \mathcal{G}\left(h_{i}\left(\alpha_{b_{i,k}}U_k+\mathcal{F}_d(d_{i,k})-P_C\right)\right)\;\;\; \text{s.t. }\\
	& \eqref{eq:p0_ltm_c1}, \eqref{eq:p0_ltm_c2},  \sum_{i=1}^{N}l_{i,k}-\phi_k\leq  0, \;  0\leq \alpha_{b_{i,k}}\leq 1,\;  d_{i,k}, l_{i,k} \geq 0,\; \phi_k\leq \tau \label{eq:p0_ltm_c3}
	\end{alignat}
\end{subequations}
for $i,m=1,\ldots,N$ and  $k=1,\ldots,K$, where \eqref{eq:p0_ltm_c1} and  \eqref{eq:p0_ltm_c2} are the energy causality and battery capacity constraints, respectively. 

$\mathrm{P_{0-LTM}}$ in \eqref{eq:p0_ltm} is non-convex due to coupling between various terms. We now transform $\mathrm{P_{0-LTM}}$ in \eqref{eq:p0_ltm} to an equivalent problem in the following.  
Define $e_{i,k}=d_{i,k}l_{i,k}$ and $\beta_{i,k}=\alpha_{b_{i,k}}l_{i,k}$.  Now, we note that the term $l_{i,k} \mathcal{G}\left(h_{i}\left(\beta_{{i,k}}U_k/l_{i,k}+\mathcal{F}_d(e_{i,k}/l_{i,k})-P_C\right)\right)$ is the perspective of $ \mathcal{G}\left(h_{i}\left(\beta_{{i,k}}U_k+\mathcal{F}_d(e_{i,k})-P_C\right)\right)$ which is a jointly concave function in $\beta_{i,k}$ and $e_{i,k}$. Since, the perspective preserves convexity, the transformed objective is a concave function \cite{Boyd}.
By the similar arguments, we note that $E_{i,k}$'s in \eqref{eq:p0_ltm_c1} are convex functions. 
Hence, $\mathrm{P_{0-LTM}}$ can be transformed to, 
\begin{subequations}\label{eq:p_ltm}
	\begin{alignat}{2}
	\mathrm{P_{LTM}}: \underset{\{\beta_{i,k}, {l_{i,k}},e_{i,k},\phi_k\}}{\text{minimize}}&\;-\sum_{k=1}^{K}\sum_{i=1}^{N}q_il_{i,k} \mathcal{G}\left(h_{i}\left(\frac{\beta_{i,k}U_k}{l_{i,k}}+\mathcal{F}_d\left(\frac{e_{i,k}}{l_{i,k}}\right)-P_C\right)\right)\;\;	\text{s.t. }\\
	&\sum_{i=1}^{m}E_{i,k} +\sum_{j=1}^{k-1}\sum_{i=1}^{N} E_{i,j}-\sum_{j=1}^{k}(\tau-\phi_j) \mathcal{F}_c(V_{a_j}^*)-B_0\leq 0 \label{eq:p_ltm_c1} \\
	&\sum_{j=1}^{k}(\tau-\phi_j) \mathcal{F}_c(V_{a_j}^*)+B_0-\sum_{i=1}^{m}E_{i,k} -\sum_{j=1}^{k-1}\sum_{i=1}^{N} E_{i,j}-B_{\mathrm{max}}\leq 0  \label{eq:p_ltm_c2}\\
	& \sum_{i=1}^{N}l_{i,k}-\phi_k\leq  0, \;  0\leq \beta_{i,k}\leq l_{i,k},\;  e_{i,k}, l_{i,k} \geq 0,\; 0\leq \phi_k\leq \tau \label{eq:p_ltm_c3}
	\end{alignat}
\end{subequations}
for $ m,i=1,\ldots, N$ and $k=1,\ldots,K$, where $E_{i,k}=\left(e_{i,k}-l_{i,k}\mathcal{F}_c(V_{i,k})\right)$, $V_{i,k}=(1-\beta_{i,k}/l_{i,k})U_k$ and  all the constraints are self-explanatory. 
In general, $\mathrm{P_{LTM}}$ in \eqref{eq:p_ltm} is non-convex due to concavity of \eqref{eq:p_ltm_c2}. When  $B_{\mathrm{max}}=\infty$, \eqref{eq:p_ltm_c2} becomes inactive and $\mathrm{P_{LTM}}$  will be convex. In the sequel, we solve $\mathrm{P_{LTM}}$ in \eqref{eq:p_ltm} for  $B_{\mathrm{max}}=\infty$ using Karush-Kuhn-Tucker (KKT) conditions, based on which we obtain the solution for arbitrary  $B_{\mathrm{max}}$ for $K=1,2$. 
For concreteness, we assume $\mathcal{G}(x)=\log(1+x)$. We present the  Lagrangian of $\mathrm{P_{LTM}}$ in \eqref{eq:p_ltm} for $B_{\mathrm{max}}=\infty$ and necessary derivatives in Appendix A. Based on \eqref{eq:db_ltm} -- \eqref{eq:dl_ltm}, we  solve  $\mathrm{P_{LTM}}$ in \eqref{eq:p_ltm} under various cases. 
\subsection{Single Frame Case}
We now consider $\mathrm{P_{LTM}}$ for $K=1$. 
\subsubsection{Ideal Case, $P_C=r=0$}
In this case, clearly, $\phi_k^*=\tau$, $\mathcal{F}_c(x)=x$, $\mathcal{F}_d(y)=y$ and, we have the following result.  
\begin{theorem}\label{lemma:TM_SF_ideal}
	For optimality, it is sufficient to transmit information in at most two layers. Let $i$ and $j\; (j>i)$ be the layers in which the information is transmitted. Then, 
	\begin{itemize}
		\item if it is optimal to exhaust the battery at the end of layer $j$, the  optimal transmit power $P^*_{m,k}=\max\left({q_m}/{\lambda^*}-{1}/{h_m},0\right)$ for $m\in\{i,j\}$,  where $\lambda^*$ is the unique solution to \eqref{eq:lambda} with $v=i$ and $w=j$,    $l_{i,k}^*=\min(\max((B_0+U_k\tau-P_{j,k}^*)/(P_{i,k}^*-P_{j,k}^*),0),\tau)$ and  $l_{j,k}^*=1-l_{i,k}^*$.
		\item if it is optimal to exhaust the battery at the end of layer $i$, $l_{i,k}^*$ is obtained from \eqref{eq:soll}, $l_{j,k}^*=1-l_{i,k}^*$, $P^*_{i,k}=U_k+B_0/l_{i,k}^*$ and $P^*_{j,k}=U_k$.  
	\end{itemize}

%	
%	
%	Further, $\alpha_{b_{m,k}}^*=1$.  Let $l_a=\min(\max((B_0+U_k-P_{j,k})/(P_{i,k}-P_{j,k}),0),\tau,B_{\mathrm{max}}/U_k)$ and $l_b=\{\max(\min(x,1),0):\left({q_ih_iB_0}\right)\left({1+h_{i}B_0+x(1+h_iU_k)}\right)-q_i\log\left(1+h_i(B_0/x+U_k)\right)+q_j\log\left(1+h_jU_k\right)=0\}$, then, 
%	\begin{equation} \label{eq:optLTMa}
%	\{P_{i,k}^*,l_{i,k}^*\}= \left\{ \,
%	\begin{IEEEeqnarraybox}[][c]{l?s}
%	\IEEEstrut
%	\{\tilde{P}_{i,k},l_a\} & if $l_a\leq\frac{B_0}{\tilde{P}_{i,k}-U_k}$,\\
%	\left\{\frac{B_0}{l_b},l_b\right\} & otherwise. 
%	\IEEEstrut
%	\end{IEEEeqnarraybox}
%	\right.
%	\end{equation}
%	$l_{j,k}^*=\tau-l_{i,k}^*$ and $e_{m,k}^*=P_{m,k}^*l_{m,k}^*$, $m=i,j$.
\end{theorem}

%\begin{theorem}\label{lemma:TM_SF_ideal}
%For optimality, it is sufficient to transmit information in at most two layers. Let $i$ and $j\; (j>i)$ be the layers in which the information is transmitted. Then, for any $m\in\{i,j\}$, $\tilde{P}_{m,k}=\max\left(\frac{q_m}{\lambda^*}-s_m,0\right)$, where, $a_{ij}=(s_i-s_j)/(q_i-q_j)$, $b_{ij}=1+1/(q_i-q_j)\left(q_i\log(s_i/q_i)-q_j\log(s_j/q_j)\right)$ and  $\lambda^*=\{\lambda:\log(\lambda)-a_{ij}\lambda+b_{ij}=0\}$.  
%Further, $\alpha_{b_{m,k}}^*=1$.  Let $l_a=\min(\max((B_0+U_k-P_{j,k})/(P_{i,k}-P_{j,k}),0),\tau,B_{\mathrm{max}}/U_k)$ and $l_b=\{\max(\min(x,1),0):\left({q_ih_iB_0}\right)\left({1+h_{i}B_0+x(1+h_iU_k)}\right)-q_i\log\left(1+h_i(B_0/x+U_k)\right)+q_j\log\left(1+h_jU_k\right)=0\}$, then, 
%\begin{equation} \label{eq:optLTMa}
%\{P_{i,k}^*,l_{i,k}^*\}= \left\{ \,
%\begin{IEEEeqnarraybox}[][c]{l?s}
%	\IEEEstrut
%	 \{\tilde{P}_{i,k},l_a\} & if $l_a\leq\frac{B_0}{\tilde{P}_{i,k}-U_k}$,\\
%	 \left\{\frac{B_0}{l_b},l_b\right\} & otherwise. 
%	\IEEEstrut
%\end{IEEEeqnarraybox}
%\right.
%\end{equation}
%$l_{j,k}^*=\tau-l_{i,k}^*$ and $e_{m,k}^*=P_{m,k}^*l_{m,k}^*$, $m=i,j$.
%\end{theorem}
\begin{proof}
	See Appendix B.
\end{proof}
A few comments are in order on Theorem \ref{lemma:TM_SF_ideal}. It is interesting to note that transmitting in two layers gives the optimal result for any channel gain distribution.
Given any two partitions, say $i$ and $j$, the optimal $\lambda^*$, $P_{i,k}^*$'s and $P_{j,k}^*$'s,  depend on the channel statistics only. Hence, $P_{i,k}^*$'s and $P_{j,k}^*$'s need to be computed only once for the given system. 
To find the optimal layers, we search across all the possible $N(N-1)/2$ combinations, taken two layers at a time. Hence, the computational complexity of solving  $\mathrm{P_{LTM}}$ in \eqref{eq:p_ltm} based on Theorem \ref{lemma:TM_SF_ideal} is  $\mathcal{O}(N^2)$ for $K=1$. 
%However, starting with layer $i$ that maximizes $q_i\log(1+h_iU_k)$, for many practical channels, we need approximately $N$ iterations to obtain the optimal solutions. 
\subsubsection{$P_C>0, r>0$}
In this case, the optimal solution is given by the following theorem. 
\begin{theorem}\label{lemma:LTM_SF_non_ideal}
For optimality, it is sufficient to transmit the information in only one layer if $0<\phi_k^*<\tau$. Whenever $\phi_k^*=\tau$, it is sufficient to transmit information in at most two layers. Algorithm \ref{algo:LTM_SF} provides an optimal solution to $\mathrm{P_{LTM}}$ in \eqref{eq:p_ltm} for $K=1$.  
\end{theorem}
\begin{proof}
	See Appendix C. 
\end{proof}

\begin{algorithm}[t]
	\caption{ {An algorithm to compute the optimal solution to $\mathrm{P_{LTM}}$ in \eqref{eq:p_ltm} when $K=1$.}}
	\label{algo:LTM_SF}
	\begin{algorithmic}[1]
		\Procedure{LTM-SF}{\textbf{$\tau,\{s_i\}_1^N,\{p_i\}_1^N, P_C, U_k,B_0, B_{\mathrm{max}}$}}
		\State Compute $\tilde{x}_{i,k}$ based on \eqref{eq:Pcphi} in Appendix C. 
		\State 	Find $\tilde{i}=\argmax_{i\in \{1,\ldots,N\}} \left(q_il_{i,k}\log(1+h_i(U_k-P_C+\mathcal{F}_d(\tilde{x}_{i,k}))) \right)$. $\alpha_{2_{i,k}}^*=1$.  
		\State Compute $\tilde{l}_{i,k}=\min\left(B_{\mathrm{max}}/\mathcal{F}_c(V_{a_k}^*),\left({B_0+\tau \mathcal{F}_c(V_{a_k}^*)}\right)/\left({\tilde{x}_{i,k}+\mathcal{F}_c(V_{a_k}^*)}\right)\right)$. 
		\If {$\tilde{l}_{\tilde{i},k}<1$} 
		\State  $\phi_k^*=l_{i,k}^*=\tilde{l}_{\tilde{i},k}$, $d_{i,k}^*=\tilde{x}_{\tilde{i},k}/l_{i,k}^*$, for $i=\tilde{i}$; $l_{i,k}^*=d_{i,k}^*=0$,
		for $i\neq \tilde{i}$. \label{Step:phi_k} 
		\Else $\;$ for each $({i,j})$ pair, compute the unique $\lambda_{m,k}$  that solves \eqref{eq:phi1} and denote it by $\lambda_k^{(i,j)}$. 
		\State $\tilde{d}_{m,k}=g_{m,k}(\lambda_k^{(i,j)})$, $\tilde{l}_{i,k}=\max\left(\min\left(\left(B_0-d_{j,k}\right)/\left(d_{i,k}-d_{j,k}\right),\tau,B_{\mathrm{max}}/\mathcal{F}_c(V_{a_k}^*)\right),0\right)$ 
		 \State for $m\in\{i,j\}$, $\tilde{l}_{j,k} =\tau-\tilde{l}_{i,k}$.
		\State Search for $(i^*,j^*)$ pair that maximizes the average rate in \eqref{eq:avgRate_LTM}. 
		\IIf {$\tilde{d}_{i^*,k}\tilde{l}_{i^*,k}>B_0$} obtain ${l}_{i^*,k}$ from \eqref{eq:lamPC} and denote it by $\tilde{l}_{i^*,k}$.
		\State $l_{i^*,k}^*= \min\left(\tilde{l}_{i^*,k},B_{\mathrm{max}}/\mathcal{F}_c(V_{a_k}^*)\right)$, $l^*_{j^*,k} =\tau-l^*_{i^*,k}$, $e_{i^*,k}^*=B_0, e^*_{j^*,k}=0$.  
		\IElse{} $\; \; l^*_{m,k}=\tilde{l}_{m,k}, d^*_{m,k}=\tilde{d}_{m,k}$ for $m=i^*,j^*$ 	\EndIIf 		\EndIf
		\State Output $l^*_{i,k}, \beta_{i,k}^*,e_{i,k}^*, \phi_k^*$ for $k=1$ and $i=1,\ldots,N$. 
		\EndProcedure
	\end{algorithmic}
\end{algorithm}
We make the following observations based on Theorem \ref{lemma:LTM_SF_non_ideal}. When $r=0$, for any finite $B_0$, $U_k$ and $P_C$, $\phi_k^*$ is strictly greater than zero. However, if $r>0$, we can have $\phi_k^*=0$ and no transmission takes place. Further, when $0<\phi_k^*<\tau$, information is transmitted in only one layer and consequently, performance using the LTM strategy is the same as the performance using the fixed rate transmission. When $\phi_k^*=\tau$, the solution is obtained based on Theorem \ref{lemma:TM_SF_ideal}. As in the ideal case, where  $P_C=r=0$, the computational complexity of Algorithm \ref{algo:LTM_SF} is $\mathcal{O}(N^2)$. 

\subsection{Multi-Frame Case}
When $B_{\mathrm{max}}=\infty$ or $r=0$,  $\mathrm{P_{LTM}}$ in \eqref{eq:p_ltm} is convex and it can be solved numerically for arbitrary K. In the sequel, we obtain the optimal solution for finite $B_{\mathrm{max}}$ and $r\geq0$ for $K=2$ in which case $\mathrm{P_{LTM}}$ in \eqref{eq:p_ltm} is non-convex. Let $\tilde{l}_{i,j}$'s, $\tilde{d}_{i,j}$'s, $\tilde{V}_{i,j}$'s and $\tilde{\phi}_j$'s be the optimal solution to $\mathrm{P_{LTM}}$ in \eqref{eq:p_ltm} with $B_{\mathrm{max}}=\infty$. In the optimal solution, the amount of energy transferred from the first frame to the second frame is given by
\begin{align}
B_1=B_0-\left(\sum_{i=1}^{N}\tilde{l}_{i,1}\tilde{d}_{i,1}-\sum_{i=1}^{N}\tilde{l}_{i,1}\mathcal{F}_c(\tilde{V}_{i,1})-(\tau-\tilde{\phi}_1) \mathcal{F}_c({V}_{a_1}^*)\right)
\end{align} 
If $B_1\leq B_{\mathrm{max}}$, \eqref{eq:p_ltm_c2} is not violated. Hence, $\tilde{l}_{i,j}$'s, $\tilde{d}_{i,j}$'s, $\tilde{V}_{i,j}$'s and $\tilde{\phi}_j$'s are optimal even for the finite $B_{\mathrm{max}}$. However, when $B_1> B_{\mathrm{max}}$, \eqref{eq:p_ltm_c2} gets violated. To account for the finite capacity of the battery in this case, we note that the rate in any frame is a concave increasing function of the initial energy in the battery. Hence, it is optimal to transfer energy from the first frame to the second frame until the battery capacity constraint is satisfied with equality, i.e., the optimal solution is obtained by solving two single frame problems - first with $B_0'=B_0-B_{\mathrm{max}}$ and second with $B_1'=B_{\mathrm{max}}$ as the initial battery energy amounts.
We present the algorithm for this case in Algorithm \ref{algo:LTM_kon_ideal}. 
\begin{algorithm}[t]
	\caption{ {An algorithm to compute the optimal solution to $\mathrm{P_{0-LTM}}$ in \eqref{eq:p0_ltm} for $K=2$}}
	\label{algo:LTM_kon_ideal}
	\begin{algorithmic}[1]
		\Procedure{LTM-non-ideal}{\textbf{$U_k,B_0, B_{\mathrm{max}},P_C,\tau,\{s_i\}_1^N,\{p_i\}_1^N,N$}}
		\State Solve $\mathrm{P_{LTM}}$ in \eqref{eq:p_ltm} and obtain $B_1$ in the optimal solution. \label{Step:infBat}
		\If {$B_1\leq B_{\mathrm{max}}$} the solution in Step \ref{Step:infBat} is optimal. 
		\Else $ $ Solve with $B_0'=B_0-B_{\mathrm{max}}$ and $B_1'=B_{\mathrm{max}}$ in each frame independently. \label{Step:finBat}
		\EndIf
		\State Output $l^*_{i,k}, \beta_{i,k}^*,e_{i,k}^*, \phi_k^*$ for $k=1,2$ and $i=1,\ldots,N$. 
		\EndProcedure
	\end{algorithmic}
\end{algorithm}
In Algorithm \ref{algo:LTM_kon_ideal}, since Step \ref{Step:infBat} and Step \ref{Step:finBat} can be solved with polynomial complexity in the worst case, we conclude that the computational complexity of Algorithm \ref{algo:LTM_kon_ideal} is polynomial. Algorithm \ref{algo:LTM_kon_ideal} is used to propose a suboptimal Algorithm later.

In the next section, we present the LSC strategy and obtain solutions in the offline case. 
%%%%%%%%%%%%%%%%%%%%%%%%%%%%%%%%%%%%%%%%%%%%%%%%%%%%%%%%%%%%%%%%%%%%%%%%%%%%%%%%%%%%%%%%%%%%%%%%%%%%%%%%%%%%%%%%%%%%%%%%%%%%%%%%%%%%%%%%%%%%%%%%
\section{LSC Strategy}\label{sec:LSC_strategy}
As in the LTM strategy, in LSC, the message is coded in $N$ layers and layer $i$ is transmitted with power $P_{i,k}$. The layers are superimposed on one another, i.e., the transmission symbol at any time is the summation of the symbols of all layers.  At the receiver, signals in the higher layers act as the interference for decoding the lower layers, hence the number of bits transmitted in layer $i$  of frame $k$ over the time duration $\phi_k$ is given by,
\begin{align}\label{eq:rate_lsc}
R_{i,k}=\phi_k\mathcal{G}\left(\frac{h_iP_{i,k}}{1+h_i\sum_{j=i+1}^{N}P_{j,k}}\right)
\end{align} 
In any frame $k$, $\text{prob}[H = h_j] = p_j$ and the achievable rate when $H = h_j$ is $\sum_{i=1}^{j}R_{i,k}$ bits/frame \cite{Goldsmith}. Hence, the average achievable rate over the channel is $\sum_{i=1}^{N}q_iR_{i,k}$, where $q_i=\sum_{j=i}^{N}p_j$. 
Without  loss of generality, we assume $P_{i,k}$'s and $P_k=\Sigma_{i=1}^{N}P_{i,k}$ remain constant over the frame. 
The total amount of energy available and consumed at the transmitter in any frame $k$ are $\phi_k\left(\alpha_{b_k}U_k+\mathcal{F}_d(d_k)\right)$ and $\phi_k\left(\sum_{i=1}^{N}P_{i,k}+P_C\right)$, respectively.  
The amount of energy stored in and drawn from the battery in frame $k$ are $(\tau-\phi_k) \mathcal{F}_c(V_{a_k}^*)+\phi_k\mathcal{F}_c(V_{k})$ and $\phi_k d_{k}$, respectively, where $V_{k}=(1-\alpha_{b_{k}})U_k$. Hence, the energy causality constraint at the transmitter and at the battery in any frame $k$ are respectively given by, 
\begin{align}
&\phi_k \left(\sum_{i=1}^{N}P_{i,k}+P_C-\alpha_{b_k} U_k-\mathcal{F}_d(d_k)\right)\leq 0 \label{eq:p0_lsc_c1} \\
&\sum_{j=1}^{k}\left(d_j\phi_j-\phi_j\mathcal{F}_c(V_{j})-(\tau-\phi_j) \mathcal{F}_c(V_{a_j}^*)\right)-B_0\leq 0  \label{eq:p0_lsc_c2} 
\end{align}
for $k=1,\ldots,K$, and the battery capacity constraint is given by, 
\begin{align}
	B_0-\sum_{j=1}^{k}\left(d_j\phi_j-\phi_j\mathcal{F}_c(V_{j})-(\tau-\phi_j) \mathcal{F}_c(V_{a_j}^*)\right)-B_{\mathrm{max}}\leq 0,\;\; \text{for}\;\; k=1,\ldots,K.  \label{eq:p0_lsc_c3} 
\end{align}
Hence, to maximize the sum rate over $K$ frames, $\mathrm{P_{\mathrm{gen}}}$ in \eqref{eq:opt_gen} can be reformulated as, 
\begin{subequations}\label{eq:p0_lsc}
	\begin{alignat}{2}
  \mathrm{P_{0-LSC}}: \; \underset{\{P_{i,k},\alpha_{b_{k}}, d_{k},\phi_k\}}{\text{maximize}}&\;\sum_{k=1}^{K}\sum_{i=1}^{N}q_iR_{i,k}\quad \text{s.t. }\\
	%&\phi_k \left(\sum_{i=1}^{N}P_{i,k}+P_C-\alpha_{b_k} U_k-\mathcal{F}_d(d_k)\right)\leq 0 \label{eq:p0_lsc_c1} \\
	%&\sum_{j=1}^{k}\left(d_j\phi_j-\phi_j\mathcal{F}_c(V_{j})-(\tau-\phi_j) \mathcal{F}_c(V_{a_j}^*)\right)-B_0\leq 0  \label{eq:p0_lsc_c2} \\
	%&B_0-\sum_{j=1}^{k}\left(d_j\phi_j-\phi_j\mathcal{F}_c(V_{j})-(\tau-\phi_j) \mathcal{F}_c(V_{a_j}^*)\right)-B_{\mathrm{max}}\leq 0  \label{eq:p0_lsc_c3} \\
	& \eqref{eq:p0_lsc_c1},\eqref{eq:p0_lsc_c2}, \eqref{eq:p0_lsc_c3}, P_{i,k}\geq 0,\; 0\leq \alpha_{b_k} \leq 1, \; 0\leq \phi_k \leq \tau,\; d_k \geq 0 \label{eq:p0_lsc_c4}
	\end{alignat}
\end{subequations}
for $k=1,\ldots,K$, $i=1,\ldots,N$, where \eqref{eq:p0_lsc_c1} and \eqref{eq:p0_lsc_c2} are the energy causality constraints and \eqref{eq:p0_lsc_c3} is the battery capacity constraint. 

Due to non-convexity of $R_{i,k}$'s, \eqref{eq:p0_lsc_c1} and \eqref{eq:p0_lsc_c2},  $\mathrm{P_{0-LSC}}$ in \eqref{eq:p0_lsc} is non-convex.  
We now transform $\mathrm{P_{0-LSC}}$ into a convex problem. 
From \eqref{eq:rate_lsc},   $\sum_{i=1}^{N}P_{i,k}=\sum_{i=1}^{N}s_i\mathcal{G}^{-1}(R_{i,k}/\phi_k)\prod_{l=1}^{i-1} \left(\mathcal{G}^{-1} \left(R_{l,k}/\phi_k\right)+1\right)$, where $s_i=1/h_i$, $s_{N+1} \overset{}{=} 0$ and  $\mathcal{G}^{-1}(x)$ is a positive, convex strictly increasing function of $x$. Since the product of non-decreasing, positive convex functions is convex \cite{Boyd}, $P_{i,k}$'s are convex functions of $R_{i,k}$'s.  
Defining $\beta_k=\alpha_{b_k}\phi_k$ and  $e_k=d_k\phi_k$, $\mathrm{P_{0-LSC}}$ in \eqref{eq:p0_lsc} can be transformed to,
\begin{subequations}\label{eq:p_lsc}
	\begin{alignat}{2}
    \mathrm{P_{LSC}}: &\; \underset{\{R_{i,k},\beta_{k}, e_{k},\phi_k\}}{\text{minimize}}\;\;-\sum_{k=1}^{K}\sum_{i=1}^{N}R_{i,k}q_i \quad \text{s.t. }\\
    & \phi_k\left( \sum_{i=1}^{N}s_i\mathcal{G}^{-1}(R_{i,k}/\phi_k)\prod_{l=1}^{i-1} \left(\mathcal{G}^{-1} \left(R_{l,k}/\phi_k\right)+1\right) +P_C -\frac{\beta_kU_k}{\phi_k} -\mathcal{F}_d\left(\frac{e_k}{\phi_k}\right)\right)\leq 0 \label{eq:p_lsc_c1} \\
	&\sum_{j=1}^{k}\left( e_j-\phi_j\mathcal{F}_c(V_{j})-B_0-(\tau-\phi_j) \mathcal{F}_c(V_{a_j}^*)\right)\leq 0 \label{eq:p_lsc_c2} \\
	&B_0+\sum_{j=1}^{k} \left(\phi_j\mathcal{F}_c(V_{j})+(\tau-\phi_j) \mathcal{F}_c(V_{a_j}^*)-e_j\right)-B_{\mathrm{max}}\leq 0 \label{eq:p_lsc_c3} \\
	&R_{i,k}, e_k\geq 0,\; 0\leq \beta_{k} \leq \phi_k,\; 0\leq \phi_k \leq \tau\label{eq:p_lsc_c4}
	\end{alignat}
\end{subequations}
for $k=1,\ldots,K$ and $i=1,\ldots,N$, where all the constraints are self explanatory. 
Noting that the perspective of a convex function is convex, we conclude that \eqref{eq:p_lsc_c1}, \eqref{eq:p_lsc_c2} and \eqref{eq:p_lsc_c3} are convex, convex and concave functions, respectively. $\mathrm{P_{LSC}}$ in \eqref{eq:p_lsc} is non-convex due to concavity of \eqref{eq:p_lsc_c2}. When  $B_{\mathrm{max}}=\infty$, \eqref{eq:p_lsc_c2} becomes inactive and $\mathrm{P_{LSC}}$ is will be convex. In the sequel, we solve $\mathrm{P_{LSC}}$ in \eqref{eq:p_lsc} for  $B_{\mathrm{max}}=\infty$ using KKT conditions, based on which we obtain the solution for arbitrary  $B_{\mathrm{max}}$ for $K=1,2$.  For concreteness, we assume $\mathcal{G}(x)=\log(1+x)$ in the rest of the section. 
We present the  Lagrangian of $\mathrm{P_{LSC}}$ in \eqref{eq:p_lsc} for $B_{\mathrm{max}}=\infty$ and necessary derivatives in Appendix D. 
Based on \eqref{eq:dR_lsc_sf} -- \eqref{eq:db_lsc_sf}, we now solve  $\mathrm{P_{LSC}}$ in \eqref{eq:p_lsc} under various cases. 

\subsection{Single Frame Case}
We consider the ideal and non-ideal cases separately, for $K=1$. 

\subsubsection{Ideal Case, $P_C=r=0$}
In this case, clearly, $\phi_k^*=\tau$, $\mathcal{F}_c(x)=x$ and $\mathcal{F}_d(y)=y$. In frame $k$, since the harvested energy is not stored in the battery, we have, $\alpha_{b_k}=1$ and $\beta_k=\tau$. From \eqref{eq:p_lsc_c2}, we have, $e_k=B_0$. Hence, $\beta_{k}$ and $e_k$ are no longer the optimization variables. The solution to $\mathrm{P_{LSC}}$ depends only on $P_k=(\beta_kU_k+e_k)/\tau$. 
Let $\lambda_k$ and $\mu_{i,k}$ be the non-negative Lagrange multipliers corresponding to \eqref{eq:p_lsc_c1} and the constraint $R_{i,k}\geq 0$ in \eqref{eq:p_lsc_c4}.  
Now, from \eqref{eq:dR_lsc_sf}, for any $i$,  ${\partial \mathrm{L_{LSC}}}/{\partial R_{i,k}}-{\partial\mathrm{L_{LSC}}}/{\partial R_{i+1,k}}=0$ implies, 
\begin{align}\label{eq:rate}
\lambda_k \exp\left(\sum_{j=1}^{i}R_{j,k}/\tau\right)=\frac{p_i}{s_i-s_{i+1}} +\frac{\mu_{i,k}-\mu_{i+1,k}}{s_i-s_{i+1}},  \quad \text{for} \;\;i=1,\ldots,N. 
\end{align}
 where we note $q_{i+1}-q_i=p_i$ and defined $\mu_{N+1,k}=0$.  Further, the complementary slackness condition requires $\mu_{i,k}R_{i,k}=0$. Hence, whenever $R_{i,k}>0$, we must have,  $\mu_{i,k}=0$.  
Now, we note that it may not be optimal to allocate the power to all the layers. To see this, assume that $R_{i,k}>0$ for $i=1,\ldots,N$. Then, due to complementary slackness condition, we must  have, $\mu_{i,k}=0,\; i=1,\ldots,N$. Since $R_{i,k}$'s are strictly positive, the left-hand side of \eqref{eq:rate} must increase with $i$. However, the right-hand side (RHS),  $p_i/(s_i-s_{i+1})$, that depends only on the channel statistics, may not always increase with $i$. This contradicts our assumption that $R_{i,k}>0$ for  $i=1,\ldots,N$, if the RHS is not increasing with $i$. 
Hence, in the following, we identify the \emph{active} layers,  the layers that are used, provided the power constraints are not violated. 

\paragraph{The Identification of the Active Layers} \label{sec:finding the active layers}
Let $\mathcal{A}$ be the set of active layers with elements arranged in ascending order of channel gains. To find $\mathcal{A}$, we adopt the technique proposed in \cite{Diggavi}.
From \eqref{eq:rate}, if all the layers are active, then $\mathcal{A}=\{1,\ldots,N\}$. 
%, $\mu_{i,k}=0, \; i=1,\ldots,N$, and the left-hand side of \eqref{eq:rate} will be  strictly increasing with $i$, implying that the right-hand side, ${p_i}/\left({s_i-s_{i+1}}\right)$, must be strictly increasing with $i$ as well.  
If ${p_j}/\left({s_j-s_{j+1}}\right) \leq {p_{j-1}}/\left({s_{j-1}-s_{j}}\right)$ for some $j\in \mathcal{A}$, then we must have $R_{j,k} \leq 0$ in order to satisfy \eqref{eq:rate}. Since $R_{j,k}$ cannot be negative, we must have, $R_{j,k}=0$. We then remove layer $j$ from $\mathcal{A}$ and update the distribution by assigning $\tilde{p}_{j-1}=p_{j-1}+p_j$ as the probability mass of $h_{j-1}$. We continue to merge the layers until ${\tilde{p}_i}/\left({s_i-s_{i+1}}\right)$ is strictly increasing with $i\in \mathcal{A}$. 

\paragraph{Rate and Power Allocation Among the Active Layers}
Let $A=|\mathcal{A}|$ be the number of active layers, indexed by $a_1,\ldots, a_A$. 
Note that $a_1$ must be $1$ and that  $h_{a_k}$ has the probability mass $\tilde{p}_{a_k}=\Sigma_{i=a_k}^{a_{k+1}-1}p_i$ for any $1\leq k \leq A$.  
Among the active layers, we have 
\begin{align}\label{eq:active_layers}
\frac{\tilde{p}_{a_1}}{s_{a_1}-s_{a_2}} <  \frac{\tilde{p}_{a_2}}{s_{a_2}-s_{a_3}} < \ldots < \frac{\tilde{p}_{a_{A-1}}}{s_{a_{A-1}}-s_{a_A}} <  \frac{\tilde{p}_{a_A}}{s_{a_A}}
\end{align}
In the optimal solution, we make the following observation.
\begin{theorem}\label{thm:LSC_SF_ideal}
	Among the active layers, power is allocated first to layer $a_A$, followed by the consecutive lower layers. The optimal power allocated to layer $a_l, \; l=1,\ldots,A$, is given by
\begin{equation} \label{eq:optLSC}
P_{a_l,k}^*= \left\{ \,
\begin{IEEEeqnarraybox}[][c]{l?s}
\IEEEstrut
P^{\mathrm{max}}_{a_l}   & if $P^{\mathrm{max}}_{a_l} \leq P_k-\sum_{j=l+1}^{A}P^{\mathrm{max}}_{a_j}$,\\
P_k-\sum_{j=l+1}^{A}P^{\mathrm{max}}_{a_j}   & otherwise.
\IEEEstrut
\end{IEEEeqnarraybox}
\right.
\end{equation}
where $P^{\mathrm{max}}_{a_1}=\infty$ and $P^{\mathrm{max}}_{a_l}$ is given by \eqref{eq:pmax} for $l=2,\ldots,A$. 
\end{theorem}
\begin{proof}
	See Appendix E.
\end{proof}

From Theorem \ref{thm:LSC_SF_ideal}, we note that all the active layers, except $a_1$, have upper limits on the optimal power allocation that depend only on the channel statistics.   Further, the highest layer will be allocated the power first. We demonstrate this solution structure in Fig. \ref{fig:SWF}. The layers are likened to containers. All the containers except the one corresponding to $a_1$ have finite capacities. The containers are arranged as shown in Fig. \ref{fig:SWF}  in a \emph{layered} manner and water (power), with volume $P_k$, is poured into the rightmost container. Note that once any container is filled,  water overflows into the immediate left container thereby emulating \eqref{eq:optLSC} in Theorem \ref{thm:LSC_SF_ideal}. We refer to the algorithm as \emph{layered water-filling} algorithm.  
Though this observation can be made from Theorem 1 in \cite{Diggavi}, it is not mentioned in \cite{Diggavi}. 
We note that, a \emph{cut-off} structure, similar to the layered water-filling structure, has been derived in a transmission completion time minimization problem in a static EH broadcast channel in \cite{BC-cutoff,BC-cutoff-finite} and a  distortion minimization problem in \cite{Goldsmith}. 

\begin{table*}[t]
	\begin{equation}\label{eq:pmax}
	P^{\mathrm{max}}_{a_l} =\left(\frac{\tilde{p}_{a_l}({s}_{a_{l-1}}-{s}_{a_l})}{\tilde{p}_{a_{l-1}}({s}_{a_l}-{s}_{a_{l+1}})}-1\right)\left({s}_{a_l}-{s}_{a_{l+1}}+ \sum_{j=l+1}^{A}\left(({s}_{a_j}-{s}_{a_{j+1}})\exp\left(\sum_{i=l+1}^{j}\log\left(\frac{\tilde{p}_{a_i}({s}_{a_{i-1}}-{s}_{a_i})}{\tilde{p}_{a_{i-1}}({s}_{a_i}-{s}_{a_{i+1}})}\right)   \right)\right)	\right)
	\end{equation}	
	\vspace{-1cm}	
\end{table*}
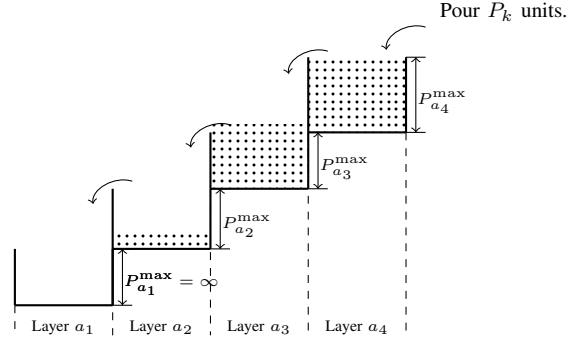
\begin{figure}[t]
	\centering
	\begin{tikzpicture} [xscale=1.3]
	\foreach \x/\lx/\y/\ly/\z in {1/0/0.75/0/1.55,2/1/.8/0.75/1.55, 3/2/.75/1.55/1.73, 4/3/1/2.3/1.3}
	{
		\begin{scope}[shift={(0,\ly)}]
		\draw [thick] (\lx,0) -- (\x,0) ; %node [below, midway] {\scriptsize $a_{\x}$};
		\draw [thick] (\lx,0) -- (\lx,\y);
		\draw [thick] (\x,0) -- (\x,\y) node [right, yshift=-0.1cm, midway] {\tiny $P_{a_{\x}}^{\mathrm{max}}$};	
		\draw  (\x,0)--(\x+0.2, 0);
		\draw (\x,\y)--(\x+0.2, \y);
		\draw [<->](\x+.1, 0)	-- (\x+0.1, \y);
		\draw[->] (\x+0.2, \z+0.1) to [bend right=60] (\x-0.2, \z-0.2);
		\end{scope}
		\draw [dashed] (\lx, \ly) -- (\lx,-0.4);
		\node at (\x-0.5,-0.3) {\tiny Layer $a_{\x}$};
	}
	\draw [dashed] (4, 2.3) -- (4,-0.4);
	\node at (5,3.9) {\scriptsize  Pour $P_k$ units.  };
	\filldraw [draw=none,pattern=dots](3,2.3) rectangle (4,3.3);
	\filldraw [draw=none,pattern=dots](2,1.55) rectangle (3,2.4);
	\filldraw [draw=none,pattern=dots](1,0.75) rectangle (2,1);
	\draw [thick] (1,0) -- (1,.75) node [right, yshift=-0.1cm, midway] {\tiny $P_{a_{1}}^{\mathrm{max}}=\infty$};	
	\end{tikzpicture}
	\caption{\scriptsize A demonstration of the \emph{layered water-filling} algorithm with four active layers. }
	\label{fig:SWF}
	%	\vspace{-1cm}	
\end{figure}
\subsubsection{Non-ideal Case}
The optimal solution is presented in the following theorem. 
\begin{theorem}\label{thm:LSC_SF_non_ideal}
	Let $\tilde{P}_{i,0}$ be the optimal solution to $\mathrm{P_{LSC}}$ in \eqref{eq:p_lsc} for $B_0=P_C\tau$ and $U_k=0$ for the given $r$.  Then, for any values of $B_0$ and $U_k$, the optimal solution to $\mathrm{P_{LSC}}$ is given by,
	\begin{align}\label{eq:LSC_opt_phi}
	\phi_k^*=\max\left\{\phi:\mathcal{F}_{d}\left(\frac{\min\left((\tau-\phi)\mathcal{F}_c\left(V_{a_k}^*\right)+B_0,B_{\mathrm{max}}\right)}{\phi}\right)+U_k-\sum_{i=1}^{N}\tilde{P}_{i,0}-P_C=0\right\}
	\end{align}
	\begin{equation}\label{eq:lsc_sf_non_ideal} 
	P_{i,k}^*= \left\{ \,
	\begin{IEEEeqnarraybox}[][c]{l?s}
	\IEEEstrut
	0  & if  $\phi_k^*=0$,\\
	\tilde{P}_{i,0}  & if  $0<\phi_k^*<\tau$,\\
	\eqref{eq:optLSC}\;\text{with}\; P_k=U_k+\mathcal{F}_d\left(\frac{B_0}{\tau}\right)-P_C & if $\phi_k^*=\tau^*$ 
	\IEEEstrut
	\end{IEEEeqnarraybox}
	\right.
	\end{equation}
   and 	$\beta_k^*=\phi_k^*$ and $e_k^*=(\tau-\phi_k^*)\mathcal{F}_c(V_{a_k}^*)+B_0$.
\end{theorem}
\begin{proof}
	See Appendix F. 
\end{proof}
We make the following remarks on Theorem \ref{thm:LSC_SF_non_ideal}. Whenever $0<\phi_k^*<\tau$, the allocated $P^*_{i,k}$'s do not depend on the specific value of $\phi_k^*$.  
%Also, whenever $0<\phi_k^*<\tau$, $\mathrm{P_{LSC}}$ in \eqref{eq:p_lsc} may have more than one solutions.  
%When the internal resistance $r=0$, the optimal $\phi_k^*=\min\left(\tau,(B_0+U_k\tau)/(P_C+\sum_{i=1}^{N}\tilde{P}_{i,k})\right)$.  
Whenever $P_C$ is finite and $r=0$, as long as the total energy available in a frame is non-zero, i.e., $B_0+U_k\tau>0$, we can always achieve a non-zero positive average rate. However, when the internal resistance is non-zero, it may be impossible to draw sufficient power to run the circuitry and power allocation may be infeasible. 
% even for a very small duration whenever $\mathcal{F}_d(B_0/\tau)+U_k\leq P_C$ and power allocation is infeasible. 
Based on the layered water-filling interpretation the optimal power allocation in Theorem \ref{thm:LSC_SF_ideal}, we note that the optimal solution to  $\mathrm{P_{LSC}}$ can be obtained in at most $N$ iterations. Hence, the computational complexity of solving $\mathrm{P_{LSC}}$ in \eqref{eq:p_lsc} for $K=1$ is $\mathcal{O}(N)$.

\subsection{Multi-frame Case}
When $B_{\mathrm{max}}=\infty$ or $r=0$, $\mathrm{P_{LSC}}$ in \eqref{eq:p_lsc} is convex and the problem can be solved for arbitrary $K$. We now solve  $\mathrm{P_{LSC}}$ in \eqref{eq:p_lsc} for arbitrary $B_{\mathrm{max}}$ and $r$ for $K=2$, when it  is non-convex.  
\subsubsection{Ideal Case, $P_C=r=0$}
%In the ideal case, $\mathrm{P_{LSC}}$ in \eqref{eq:p_lsc} can be readily solved based on the following observation. 
In this case, we first make the following important observation. 
\begin{lemma}\label{lemma:LSC_MF_ideal}
	The optimal average rate within any frame $k$, denoted by $R_{k}^{(LSC)}(P_k)$, obtained by solving $\mathrm{P_{LSC}}$ in \eqref{eq:p_lsc}, is a concave increasing function of the  uniform transmit power $P_k$.  
\end{lemma}
\begin{proof}
In this case, \eqref{eq:p_lsc_c1} can be re-written as $\sum_{i=1}^{N}P_{i,k}-P_k\leq0$, where $P_k=U_k+B_0/\tau$. 
In any frame $k$, as $P_k$ increases, the constraint \eqref{eq:p_lsc_c1} gets relaxed, or, in other words, the constraint \eqref{eq:p_lsc_c1} is perturbed. Since, the optimal value function of a perturbed problem is convex if the original problem is convex (exercise $5.32$ in \cite{Boyd}), the result follows.  
\end{proof}
Hence, in order to find the optimal energy allocation to each of the frames, we need to solve the following convex optimization problem.
\begin{subequations}\label{eq:p_lsc_mf_ideal}
	\begin{alignat}{2}
	\underset{\{P_k\}}{\text{minimize}}\;\;&-\sum_{k=1}^{K}R_{k}^{(LSC)}(P_k) \\
	\text{s.t. } \;\;& \sum_{j=1}^{k}(P_j-U_j)\tau-B_0\leq 0, \; P_k \geq 0,\; k = 1, \ldots, K \label{eq:p_lsc_mf_ideal_c1}
	\end{alignat}
\end{subequations} 
Note that this problem is a specific case of the general problem solved in \cite{yang}. Based on \cite{yang}, the optimal solution has the following properties. 
\begin{lemma}\label{eq:p_lsc_mf_ideal1}
	Optimal $P_k$'s form a non-decreasing sequence, i.e., $P_1^*\leq P_2^*\leq,\ldots,P_K^*$ and whenever $P_k^*$'s change the value, the entire harvested energy up to that frame is consumed, i.e.,  for any $k$,  $P_k^*<P_{k+1}^*$ implies $\sum_{k=1}^{j}(U_k-P_k^*)\tau-B_0=0$.  
\end{lemma}
\begin{proof}
	The result follows due to the concavity of $\mathcal{R}^*(\cdot)$ and it can be proved along the lines of the proofs of Lemma 1 and Lemma 3 in \cite{yang}. 
\end{proof}
From $P_k^*$'s, we can easily compute $\alpha_{b_k}^*$'s and $d_k^*$'s, and $P_{i,k}^*$'s can be found from Theorem \ref{thm:LSC_SF_ideal}. Since the set of active layers depends only on the channel statistics, it needs to be computed only once for the given system. Due to the non-decreasing structure of optimal power levels across the frames and  the layered water-filling structure within a frame, in frame $k$, if any active layer $a_l$ (hence, $a_{l+1}$ to $a_A$) is allocated power up to $P^{\mathrm{max}}_{a_l}$, then layers $a_l$ to  $a_A$ are also allocated power up to their thresholds in frames $k+1,\ldots, K$. Hence, the power needs to be computed only for layers up to $a_l$ and the search space reduces significantly. Further, we can easily account for the finite $B_{\mathrm{max}}$ in \eqref{eq:p_lsc_mf_ideal} as in \cite{ulukusFiniteBat}.
From \cite{yang}, we note that \eqref{eq:p_lsc_mf_ideal} can be solved in $K$ iterations. Hence, the computational complexity in solving $\mathrm{P_{LSC}}$ in \eqref{eq:p_lsc} in the ideal case  is $\mathcal{O}(NK)$. 
\begin{algorithm}[t]
	\caption{ {An algorithm to compute the optimal solution to $\mathrm{P_{LSC}}$ in \eqref{eq:p_lsc} for $K=2$}}
	\label{algo:LSU_kon_ideal}
	\begin{algorithmic}[1]
		\Procedure{LSC-non-ideal}{\textbf{$U_k,B_0, B_{\mathrm{max}},P_C,\tau,\{s_i\}_1^N,\{p_i\}_1^N,N$}}
		\State Obtain $\tilde{P}_{i,0}$ with $B_0=P_C\tau$, $U_k=0$, and $\tilde{\phi}_k$ from \eqref{eq:LSC_opt_phi} for $k=1,2$ as in Theorem \ref{thm:LSC_SF_non_ideal}. \label{Step:opt_phi}
		\If {$\tilde{\phi}_1<\tau$}  solve for $P^*_{i,k}$'s from Theorem \ref{thm:LSC_SF_non_ideal} for $k=1,2$ independently.  
		\ElsIf {$\tilde{\phi}_1=\tau$ and $\tilde{\phi}_2<\tau$} obtain $e_1$ using \eqref{eq:d2e1}.
		\IIf {$e_1>0$} $e_1^*=e_1$ and  obtain $\phi_2^*$ using \eqref{eq:phi2} and $P_1^*=\mathcal{F}_d(e_1^*/\tau)+U_1-P_C$. 
		\IElse{} $\;$ obtain $\beta_1^*$ using \eqref{eq:case3} and $P_1^*=\beta_1^*U_1/\tau$. 	\EndIIf
		\State Solve for $P^*_{i,1}$'s from Theorem \ref{thm:LSC_SF_non_ideal} with $P_1=P_1^*$ and  $P_{i,2}^*=\tilde{P}_{i,0}$, $i=1,\ldots,N$. 
		\ElsIf {$\phi_1=\tau$ and $\phi_2=\tau$} obtain $\beta_1^*$ using \eqref{eq:case3} with $d_2=B_0/\tau+\mathcal{F}_c(V_1)$. 
		\State Solve for $P^*_{i,k}$'s from Theorem \ref{thm:LSC_SF_non_ideal} with $P_1=\beta_1^*U_1/\tau$, $P_2=\mathcal{F}_d(d_2^*)+U_2-P_C$. 
		\EndIf
		\State Output $\phi_k^*$, $\beta_{i,k}^*,e_{i,k}^*,R_{i,k}^*$ for $k=1,2$ and $i=1,\ldots,N$. 
		\EndProcedure
	\end{algorithmic}
\end{algorithm}
\subsubsection{Non-ideal Case, $P_C\geq 0$ and $r\geq 0$}
We solve $\mathrm{P_{LSC}}$ in \eqref{eq:p_lsc} when $K=2$ based on which we propose an online algorithm. 
%For the higher values of $K$, the analysis becomes unwieldy and we only provide numerical solutions. 
We have the following result in this case. 
\begin{theorem}\label{thm:LSC_MF_non_ideal}
The Algorithm \ref{algo:LSU_kon_ideal} gives the optimal solution to $\mathrm{P_{LSC}}$ in \eqref{eq:p_lsc} for $K=2$. 
\end{theorem}
\begin{proof}
	See Appendix G. 
\end{proof}
Since for any given $d_k$ and $\beta_k$, $R_{j,k}$'s can be obtained from Theorem \ref{thm:LSC_SF_ideal}, we can obtain $e_1$, $\beta_1^*$ in $N$ iterations. Note that $\tilde{P}_{i,0}$ needs to be computed only once for the given system with polynomial complexity. Assuming that $\tilde{P}_{i,0}$'s are known, each of the steps in Algorithm \ref{algo:LSU_kon_ideal} requires at most $N$ iterations. Hence, the computational complexity for the two frame case is  $\mathcal{O}(N)$.  

\section{Online Policies} \label{statistical}
In practice, it would be unrealistic to have the non-causal knowledge of the harvested power, but, it is likely that we have statistical information.
We now present the optimal online policy, a suboptimal online policy inspired by the offline solution, and a greedy policy in this section. 

\subsubsection{Optimal Online Policy}
To obtain the optimal power allocation when only the causal knowledge and the statistical information of the harvested powers are available, we employ the stochastic dynamic programming based approach \cite{dp}. We describe the problem formulation for the LSC strategy only. The similar approach can be used to formulate the problem using the LTM strategy.  
Let $\zeta_k=(U_k,B_{k-1})$ denote the state of the system in frame $k$, where $U_k$ is the harvested power and $B_{k-1}$ is the energy available in the battery at the start of the frame $k$. We assume that the state information of any given frame is known at the start of the frame. Note that $\zeta_1=(U_1,B_0)$ is the initial state of the system. 	Our goal is to maximize the sum rate over a finite horizon of $K$ frames, by choosing a policy, $\pi=\{P_{i,k}(\zeta_k),\phi_k(\zeta_k),\alpha_{b_k}(\zeta_k),d_{k}(\zeta_k),\forall \zeta_k, k=1,\ldots,K,\; i=1,\ldots,N\}$, that selects  power allocation (to each of the layers),  transmission duration,  power splitting ratios and discharge powers for each of the frames. A policy is feasible if the energy causality constraints and  battery capacity constraints specified in \eqref{eq:p0_lsc_c1} -- \eqref{eq:p0_lsc_c3}, are satisfied for possible states in all the frames. Let $\Pi$ denote the set of all feasible policies. Given the initial state $\zeta_1$, the maximum average rate is given by,
$\mathcal{R}_{\mathrm{on}}^*=\max_{\pi\in\Pi}\mathcal{R}_{\mathrm{on}}(\pi)$, 
where 
\begin{align}\label{eq:Ron}
\mathcal{R}_{\mathrm{on}}(\pi)=\sum_{k=1}^{K}\mathbb{E}\left[R(U_k,B_{k-1},P_{i,k}(\zeta_k),\phi_k(\zeta_k),\alpha_{b_k}(\zeta_k),d_{k}(\zeta_k))|\zeta_1,\pi\right]
\end{align}
with $R=\sum_{i=1}^{N}q_iR_{i,k}$, where $R_{i,k}$ is given by \eqref{eq:rate_lsc}. The expectation in \eqref{eq:Ron} is with respect to the random harvested power. 
The maximum average rate, $\mathcal{R}_{\mathrm{on}}^*$ of the system, given by the value function $J_1(\zeta_1)$, can be computed recursively based on Bellman's equations, starting from $J_K(\zeta_K), J_{K-1}(\zeta_{K-1})$, and so on until $J_1(\zeta_1)$ as follows:
	\begin{subequations}\label{eq:dp}
		\begin{alignat}{2}
		&J_K(U_K,B_{K-1})=\max R(U_K,B_{K-1},P_{i,k}(\zeta_k),\phi_k(\zeta_k),\alpha_{b_k}(\zeta_k),d_{k}(\zeta_k)) \label{eq:dpN}\\
		&J_k(U_k,B_{k-1})=\max R(U_k,B_{k-1},P_{i,k}(\zeta_k),\phi_k(\zeta_k),\alpha_{b_k}(\zeta_k),d_{k}(\zeta_k))+\bar{J}_{k+1}(U_{k+1},B_k)\nonumber\\ 
		&\qquad\qquad\qquad\qquad\qquad\qquad\qquad\qquad\text{for}\;k=1,\ldots,K-1, \;i=1,\ldots,N  \label{eq:dpn}
		\end{alignat}
	\end{subequations}
where the maximization in \eqref{eq:dpN} and \eqref{eq:dpn} is over $\{P_{i,k}(\zeta_k),\phi_k(\zeta_k),\alpha_{b_k}(\zeta_k),d_{k}(\zeta_k)\}$ and $\bar{J}_{k+1}(U_{k+1},x)=\mathbb{E}_{U_{k+1}}\left[{J}_{k+1}(U_{k+1},x)\right]$ is the average throughput across frames $k+1$ to $K$ averaged over all the realizations of $U_{k+1}$. Note that in \eqref{eq:dpn}, we account for the fact that $U_i$'s are independent.  Note that the residual energy $B_k$ in \eqref{eq:dpn} is a function of the decision variables $P_{i,k}(\zeta_k),\phi_k(\zeta_k),\alpha_{b_k}(\zeta_k),d_{k}(\zeta_k)$. An optimal policy is denoted as $\pi^*=\{P_{i,k}^*(\zeta_k),\phi_k^*(\zeta_k),\alpha_{b_k}^*(\zeta_k),d_{k}^*(\zeta_k), \forall \zeta_k, k=1,\ldots,K, \;i=1,\ldots,N \}$, where the optimal solution to \eqref{eq:dp} is given by $\{P_{i,k}^*(\zeta_k),\phi_k^*(\zeta_k),\alpha_{b_k}^*(\zeta_k),d_{k}^*(\zeta_k)\}$ when the state of the system is $\zeta_k$. 

\subsubsection{Mean Value Based (MV) Policy}
In addition to the instantaneous knowledge, when we have the statistical information (such as the mean value) of harvested powers, we propose an algorithm for LSC strategy based on Algorithm \ref{thm:LSC_MF_non_ideal}. Let the expected values of the harvested power be $\bar{U}$. Then, MB policy works as follows. 
At the beginning of any frame $k$, we have knowledge of the harvested power $U_k$, residual energy in the battery, $B_{k-1}$. To find $(P_{i,k},\phi_k,\alpha_{b_k},d_{k})$, we consider a hypothetical two-frame optimization problem with the first frame being the frame $k$ and the second frame being a hypothetical frame with harvested power $0.5 \bar{U}$. Then, at the beginning of frame $k$, $k=1,\ldots, K$, the transmitter solves the optimization problem  $\mathrm{P_{LSC}}$ in \eqref{eq:p_lsc} for the above two-frame hypothetical problem.   The residual energy in the battery is considered as the initial energy stored in the battery for the next iteration. The similar algorithm can be obtained for the LTM strategy in which $\mathrm{P_{LTM}}$ in \eqref{eq:p_ltm} is solved for the above two-frame problem. 
\subsubsection{Greedy Algorithm}
When we only have the instantaneous knowledge of the harvested power but not the non-causal or statistical information on the power profile, entire harvested energy in any frame is utilized in the same frame itself. In each of the frames, the corresponding single frame optimization problem is solved based on Algorithm \ref{algo:LTM_SF} and Theorem \ref{thm:LSC_SF_non_ideal},  for LTM and LSC strategies, respectively.

\section{Numerical Results}\label{sec:numerical results}
Based on \cite{IR}, we assume $\mathcal{F}_c(V)=-rV^2/V_B^2+V$ and  $\mathcal{F}_d(d)=-rd^2/V_B^2+d$, where $V_B $ is the nominal voltage of the battery. 
We assume $\mathcal{G}(z)=W\log(1+z/(N_0W))$, where  $z\;\si{\watt}$ is the transmit power, $W=\SI{1}{\mega\hertz}$ is the channel bandwidth and $N_0=1$ \si{\nano\watt/\hertz}.
We assume the power gain $H$ with Gamma distribution: $f_H(h;\,x,y) = y^x h^{x-1}e^{-yh}/\Gamma(x)$, 
where $x$ is the shape parameter, $y$ is the scale parameter and $\Gamma(x)$ is the Gamma function. We truncate  $f_H(h)$ at $h=T$ and quantize $h$ to $N$ evenly spaced levels in $[0,T]$ obtaining $h_i=iT/N$ with probability $p_i=\int_{h=({i-1})T/N}^{i T/N}f_H(h)dh$ for $i=1,\ldots,N-1$, $h_N=T/N$ and  $p_N=\int_{h=({N-1})T/N}^{\infty}f_H(h)dh$. 
 \begin{figure*}[t]
	\centering
	\begin{subfigure}{0.48\textwidth}
		\includegraphics[width=1\textwidth]{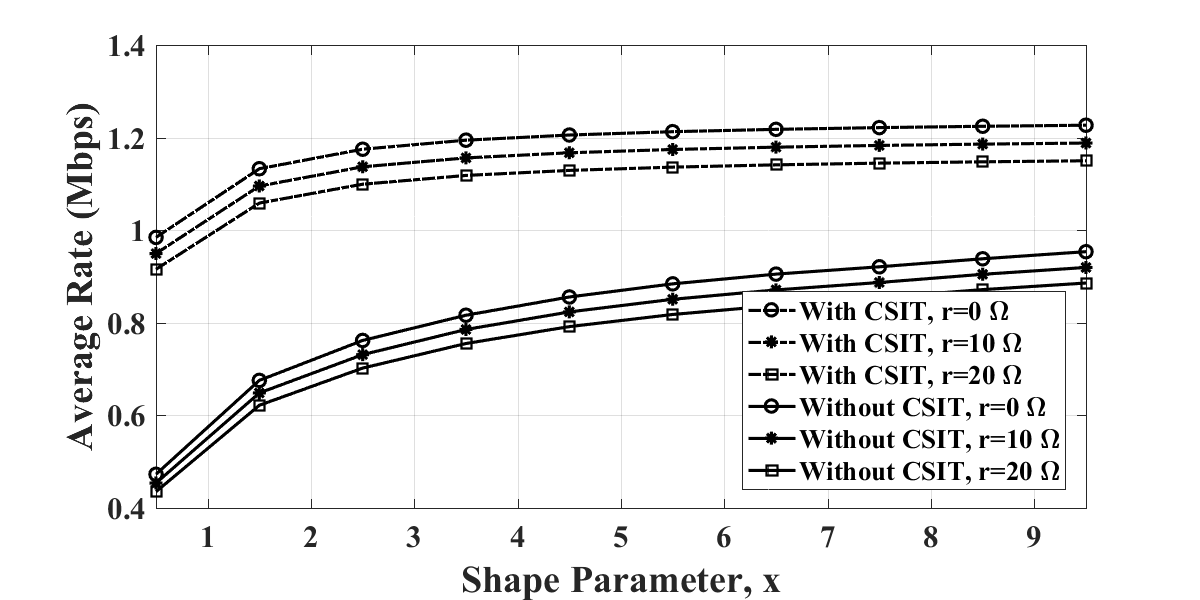}
		\caption{LTM strategy }
		\label{fig:ltm_var_shape}
	\end{subfigure}%
	\begin{subfigure}{0.48\textwidth}
		\includegraphics[width=1\textwidth]{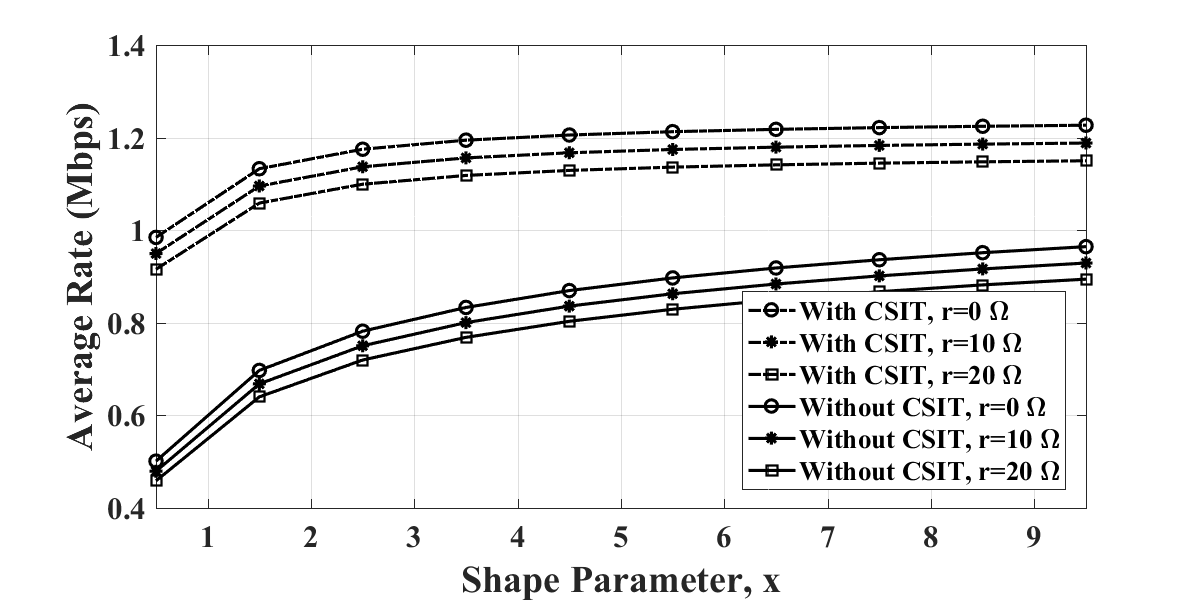}
		\caption{LSC strategy }
		\label{fig:lsc_var_shape}
	\end{subfigure}
	\caption{Variation of optimal average rate with shape parameter, $x$, for $K=1$, $y=1/x$ such that $\sum_{i=1}^{N}h_ip_i=1$, $T=5$, $V_B=1.5$ \si{\volt}, $B_{\mathrm{max}}=30$ \si{\milli\watt} and $U_1=P_C=10$ \si{\milli\watt}.   }
	\label{fig:var_shape}
\end{figure*}

\begin{figure*}[t]
	\centering
	\begin{subfigure}{0.4\textwidth}
		\includegraphics[width=1\textwidth]{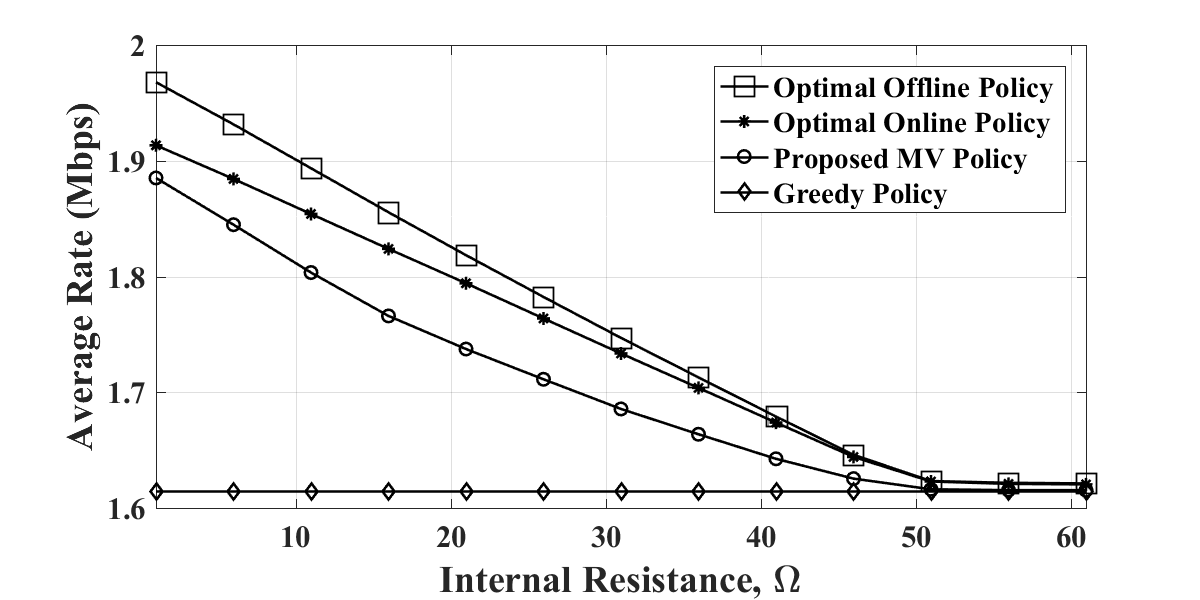}
		\caption{LTM strategy }
		\label{fig:ltm_var_r}
	\end{subfigure}%
\hspace{1cm}
	\begin{subfigure}{0.4\textwidth}
		\includegraphics[width=1\textwidth]{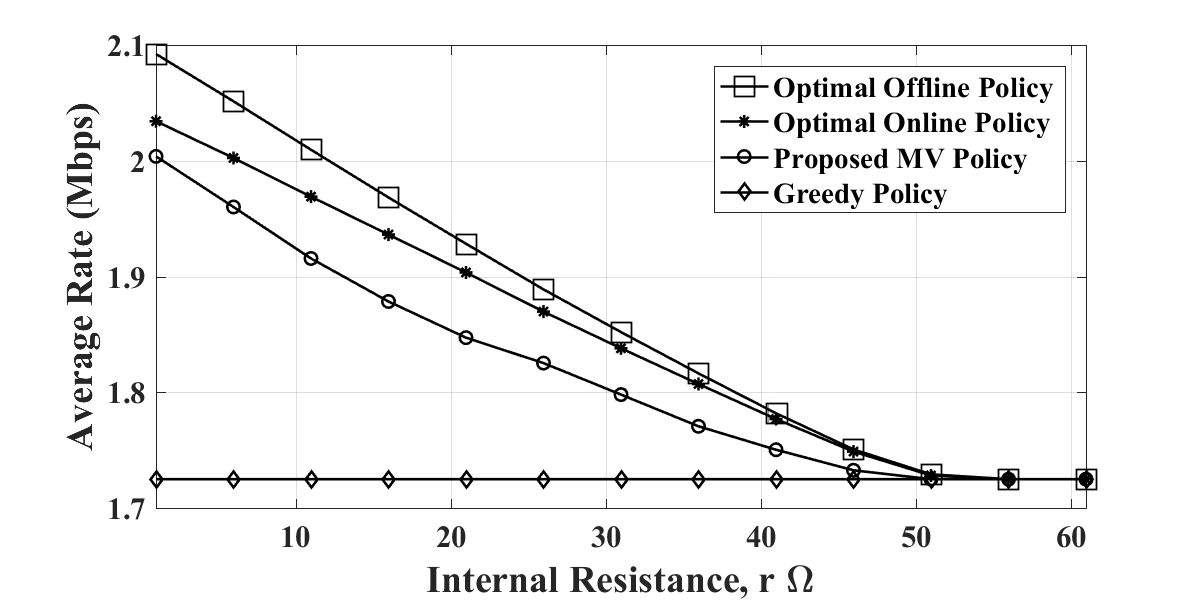}
		\caption{LSC strategy }
		\label{fig:lsc_var_r}
	\end{subfigure}
	\caption{Variation of the optimal rate with the internal resistance for $x=y=1$, $T=5$, $V_B=1.5$ \si{\volt}, $B_{\mathrm{max}}=\infty$, $P_C=10$ \si{\milli\watt} and $U$, uniformly distributed over $\{0,50,100\}\;\si{\milli\watt}$.}   
	\label{fig:var_r}
\end{figure*}

In Fig. \ref{fig:var_shape}, we present the variation of the average rate with the shaping parameter, $x$, with and without CSIT for a fixed mean value for $K=1$.  As $x$ increases, the channel becomes more deterministic, i.e., the probability of a particular channel realization dominates all others. From Fig. \ref{fig:ltm_var_shape} and Fig. \ref{fig:lsc_var_shape},  we note the average rates in both the cases increase with the shaping parameter $x$ and the performance without CSIT using the layered coding approaches the performance with CSIT. 
%We also note that the internal resistance reduces the power available and, it affects the performance in all the cases. 

In Fig. \ref{fig:var_r}, we present the variation of the average rate for $K=50$ frames with the internal resistance for the offline and online policies for LTM and LSC strategies with $B_{\mathrm{max}}=\infty$. The offline optimal results are obtained by solving $\mathrm{P_{LTM}}$ in \eqref{eq:p_ltm} and $\mathrm{P_{LSC}}$ in \eqref{eq:p_lsc}. 
As expected, the average rate in LTM strategy is lower than the LSC strategy  always.  
In all the policies, except the Greedy policy, the average rate decreases with the internal resistance and meets the performance of the Greedy policy when the internal resistance is high. This is because the losses across the internal resistance prohibits  energy transfer across the frames. The average rate in the Greedy policy does not depend on the internal resistance because, in each of the frames,  it is optimal to not store energy in the battery due to battery losses.  Also, the proposed MV policy performs significantly better than the Greedy policy when the internal resistance is small.  

In Fig. \ref{fig:var_Bmax}, we present the variation of the average rate for $K=50$ frames with the battery capacity for the offline and online policies for LTM and LSC strategies. Since the offline optimization problems are non-convex with the finite capacity battery, we obtain the results using dynamic programming.  In all the policies,  except the Greedy policy, the average rate increases with the increasing capacity of the battery. After a certain value the battery capacity  the rate of increment of the average rate reduces significantly and reaches a plateau beyond which the battery capacity does not play any role. Note that in the Greedy policy, there is no change in the average rate with the battery capacity as the energy is not stored in the battery. 

\begin{figure*}[t]
	\centering
	\begin{subfigure}{0.4\textwidth}
		\includegraphics[width=1\textwidth]{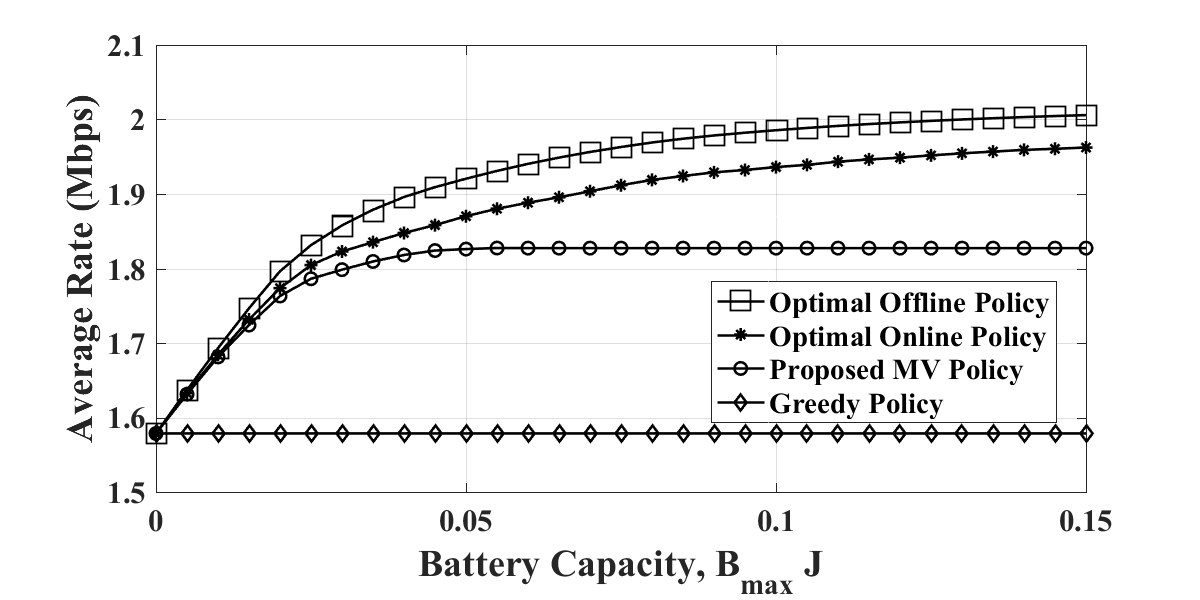}
		\caption{LTM strategy }
		\label{fig:ltm_var_Bmax}
	\end{subfigure}%
\hspace{1cm}
	\begin{subfigure}{0.4\textwidth}
		\includegraphics[width=1\textwidth]{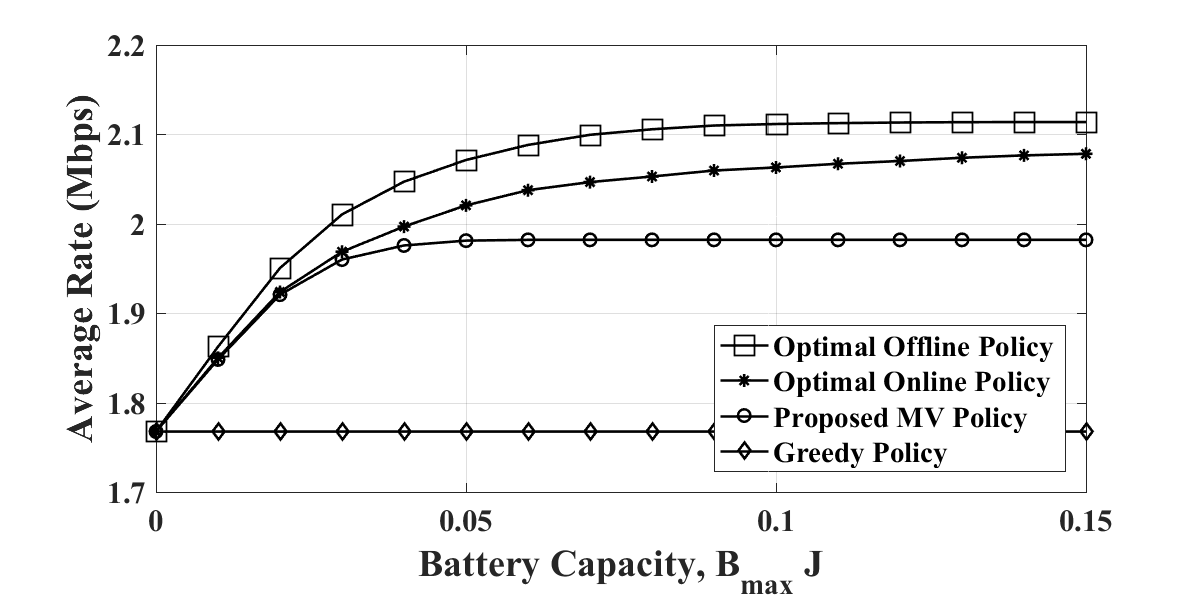}
		\caption{LSC strategy }
		\label{fig:lsc_var_Bmax}
	\end{subfigure}
	\caption{Variation of the optimal rate with the maximum capacity of the battery for $x=y=1$, $T=5$, $V_B=1.5$ \si{\volt}, $r=5$ \si{\ohm}, $P_C=10$ \si{\milli\watt} and $U$, uniformly distributed over $\{0,50,100\}\;\si{\milli\watt}$.}   
	\label{fig:var_Bmax}
\end{figure*}

\section{Conclusions}\label{sec:reflections}
In this paper,  we noted that it may be infeasible to acquire the current CSIT in EH  communication systems due to stringent constraints on resources. 
We optimized an EH transmitter communicating over a slow fading channel, which has access to the channel statistics, but does not know the exact channel state, under layered time-multiplexing and layered superposition coding strategies. In general, we have shown that the average rate maximization  problems are non-convex, and we reformulated and solved the problems  for the two frame case in the offline setting. We then proposed  heuristic online algorithms  based on the offline solutions and showed that the proposed algorithms perform significantly better than the greedy policies. For the superposition coding, we provided a simple and concise interpretation, referred to as layered water-filling algorithm, for the optimal solution in an ideal single frame case. By numerical simulations, we noted that the internal resistance significantly affects the system performance, and the optimal performance of the offline and online policies approach the performance of the greedy policy as the internal resistance increases. 

\section*{Appendix}
\subsection{The Lagrangian of  $\mathrm{P_{LTM}}$ in \eqref{eq:p_ltm} when $B_{\mathrm{max}}=\infty$ and Necessary Derivatives}
The Lagrangian of $\mathrm{P_{LTM}}$ in \eqref{eq:p_ltm} when $B_{\mathrm{max}}=\infty$ is given by
\begin{align}\label{eq:lagrange_LTM}
L_{\mathrm{LTM}}=& \sum_{k=1}^{K}\left(-\sum_{i=1}^{N}q_il_{i,k} \mathcal{G}\left(h_{i}\left(\frac{\beta_{i,k}U_k}{l_{i,k}}+\mathcal{F}_d\left(\frac{e_{i,k}}{l_{i,k}}\right)-P_C\right)\right)+ 
\mu_k\left(\sum_{i=1}^{N}l_{i,k}-\phi_k\right)\right)\nonumber\\
&+\sum_{k=1}^{K}\left(\sum_{m=1}^{N} \lambda_{m,k} \left(\sum_{i=1}^{m} \left(e_{i,k}-l_{i,k}\mathcal{F}_c(V_{i,k})\right)-(\tau-\phi_k) \mathcal{F}_c(V_{a_k}^*)\right)-\sum_{i=1}^{N}\nu_{l_{i,k}}l_{i,k}\right)+\nonumber\\
&\sum_{k=1}^{K}\sum_{m=1}^{N} \lambda_{m,k} \left(\sum_{j=1}^{k-1}\left(\sum_{i=1}^{N} \left(e_{i,j}-l_{i,j}\mathcal{F}_c(V_{i,j})\right)-(\tau-\phi_j) \mathcal{F}_c(V_{a_j}^*)\right)-B_0\right)+\nonumber\\
&\sum_{k=1}^{K}\left(\sum_{i=1}^{N}\omega_{\beta_{i,k}}(\beta_{i,k}-l_{i,k})-\sum_{i=1}^{N}\nu_{\beta_{i,k}}\beta_{i,k}
-\sum_{i=1}^{N}\nu_{e_{i,k}}e_{i,k}-\nu_{\phi_{k}}\phi_{k}+\omega_{\phi_{k}}(\phi_{k}-\tau)\right)
\end{align}
where $\lambda_{m,k}, \mu_k, \nu_{\beta_{i,k}}, \omega_{\beta_{i,k}},\nu_{e_{i,k}}, \nu_{l_{i,k}}, \nu_{\phi_{k}},\omega_{\phi_{k}}$ are non-negative Lagrange multipliers corresponding to \eqref{eq:p_ltm_c1}, and constraints in \eqref{eq:p_ltm_c2}, i.e., $\sum_{i=1}^{N}l_{i,k}-\phi_k\leq  0, \;  -\beta_{i,k}\leq 0,\; \beta_{i,k}\leq l_{i,k},\; -e_{i,k}\leq 0, -l_{i,k} \leq 0,\;-\phi_k\leq 0$ and $\phi_k\leq \tau$, respectively. 

Let $f'(x)=\partial f(x)/\partial x$.  Differentiating $\mathrm{L_{LTM}}$ in \eqref{eq:lagrange_LTM} with respect to $\beta_{i,k}, e_{i,k}, l_{i,k}, \phi_k$, we have,
\begin{align}
&\frac{\partial L_{\mathrm{LTM}}}{\partial \beta_{i,k}}=\frac{-h_iq_iU_k}{1+h_{i}P_{i,k}}+U_k\mathcal{F}_c'\left(V_{i,k}\right)\left(\sum_{j=k+1}^{K}\sum_{m=1}^{N}\lambda_{m,j}+\sum_{m=i}^{N}\lambda_{m,k}\right)-\nu_{\beta_{i,k}}+\omega_{\beta_{i,k}}=0 \label{eq:db_ltm}\\
&\frac{\partial L_{\mathrm{LTM}}}{\partial e_{i,k}}= \frac{-h_iq_i\mathcal{F}_d'\left(\frac{e_{i,k}}{l_{i,k}}\right)}{1+h_{i}P_{i,k}}+\left(\sum_{j=k+1}^{K}\sum_{m=1}^{N}\lambda_{m,j}+\sum_{m=i}^{N}\lambda_{m,k}\right)-\nu_{e_{i,k}}=0\label{eq:de_ltm}\\
&\frac{\partial L_{\mathrm{LTM}}}{\partial l_{i,k}}= \frac{h_iq_i\left(\frac{\beta_{i,k}U_k}{l_{i,k}}+\frac{e_{i,k}}{l_{i,k}}\mathcal{F}_d'\left(\frac{e_{i,k}}{l_{i,k}}\right)\right)}{1+h_{i}P_{i,k}}-q_i\log\left(1+h_iP_{i,k}\right)+\mu_k-\nu_{l_{i,k}}-\omega_{\beta_{i,k}}+\nonumber\\
&\qquad\qquad\left(-\frac{\beta_{i,k}U_k}{l_{i,k}}\mathcal{F}_c'\left(V_{i,k}\right)-\mathcal{F}_c\left(V_{i,k}\right)\right)\left(\sum_{j=k+1}^{K}\sum_{m=1}^{N}\lambda_{m,j}+\sum_{m=i}^{N}\lambda_{m,k}\right)\label{eq:dl_ltm}\\
&\frac{\partial L_{\mathrm{LTM}}}{\partial \phi_k}=-\mu_k+\mathcal{F}_c(V_{a_k}^*)\left(\sum_{j=k+1}^{K}\sum_{m=1}^{N}\lambda_{m,j}+\sum_{m=i}^{N}\lambda_{m,k}\right)-\nu_{\phi_{k}}+\omega_{\phi_{k}}=0\label{eq:dphi_ltm}
\end{align}

\subsection{Proof of Theorem \ref{lemma:TM_SF_ideal} }
In the ideal case, the harvested energy is not stored in the battery. Hence, \eqref{eq:p_ltm_c2} is inactive. We note that the battery can get exhausted in any of the partitions. 
We first consider the case when the battery is exhausted in partition $m$, the last partition in which the information is transmitted, i.e., when \eqref{eq:p_ltm_c1} is satisfied with equality only for $j=m$. In this case, $\lambda_{j,k}=0$ for $j=1,\ldots, m-1$ and $\lambda_{m,k}\geq 0$.  
For any layer $i\;(\leq m)$ in which the information is transmitted, we have, $P_{i,k}=\beta_{i,k}U_k/l_{i,k}+e_{i,k}/l_{i,k}>0$, $l_{i,k}>0$,  and either or both $\beta_{i,k}$ and $e_{i,k}$ must be non-zero. This implies, $\nu_{\beta_{i,k}}=\omega_{\beta_{i,k}}=\nu_{l_{i,k}}=0$, due to complementary slackness conditions. Consequently, whenever $P_{i,k}>0$, from \eqref{eq:db_ltm}, \eqref{eq:de_ltm} and \eqref{eq:dl_ltm}, we get, 
\begin{align}
\lambda_{m,k}&=\frac{h_iq_i}{1+h_iP_{i,k}}\implies P_{i,k}=\frac{q_i}{\lambda_{m,k}}-s_i, \;\; \forall\; i\in\mathcal{A}_T \label{eq:lam}\\
\mu_k&=-\frac{h_iq_i P_{i,k}}{1+h_iP_{i,k}}+q_i\log(1+h_iP_{i,k})+\lambda_{m,k}U_k , \;\; \forall\; i\in\mathcal{A}_T \label{eq:om}
\end{align}
where $s_i=1/h_i$ and $\mathcal{A}_T\subseteq\{1,\ldots, m\}$, is the set of layers in which the information is transmitted. 
By substituting the expression for $P_{i,k}$ from \eqref{eq:lam} in \eqref{eq:om}, we have, 
\begin{align}\label{eq:om_lam}
\mu_k=-q_i\log(\lambda_{m,k})+\lambda_{m,k} (s_i+U_k)-q_i-q_i \log\left(\frac{s_i}{q_i}\right), \;\; \forall\; i\in\mathcal{A}_T
\end{align}
Note that \eqref{eq:om_lam} is a system of $|\mathcal{A}_T|$ non-linear equations with two unknowns, $\lambda_{m,k}$ and $\mu_k$, where $|\mathcal{A}_T|$ is the cardinality of $\mathcal{A}_T$. When $|\mathcal{A}_T|\leq 2$, the solution to \eqref{eq:om_lam} can be easily found. However, when $|\mathcal{A}_T|> 2$, \eqref{eq:om_lam} is an overdetermined system of equations. In such cases, we now show that \eqref{eq:om_lam} is inconsistent. 
Let $\mathcal{A}_T=\{{i,j,m}\},\; i\neq j\neq m$ such that $s_i>s_j>s_m$, without loss of generality. Now, from \eqref{eq:om_lam}, for layers $i$, $j$ and $m$, we obtain the following equations.  
\begin{align}\label{eq:lambda}
 \log(\lambda_{m,k})-a_{ij}\lambda_{m,k}-b_{ij}= \log(\lambda_{m,k})-a_{im}\lambda_{m,k}-b_{im}=\log(\lambda_{m,k})-a_{jm}\lambda_{m,k}-b_{jm}=0
\end{align}
where $a_{vw}=(s_v-s_w)/(q_v-q_w)$ and $b_{vw}=-1-1/(q_v-q_w)\left(q_v\log(s_v/q_v)-q_w\log(s_w/q_w)\right)$ for $v=i,i,j$ and $w=j,m,m$, respectively. 
From the first two equations in \eqref{eq:lambda}, we have, $\tilde{\lambda}_{m,k}=({b_{im}-b_{ij}})/({a_{ij}-a_{im}})$ and, 
\begin{align}\label{eq:verify}
\log(\tilde{\lambda}_{m,k})=&\log\left(\frac{b_{im}-b_{ij}}{a_{ij}-a_{im}}\right) \stackrel{\text{(a)}}{<} \frac{b_{im}-b_{ij}}{a_{ij}-a_{im}}\stackrel{\text{(b)}}{<} \frac{b_{im}a_{ij}-b_{ij}a_{im}}{a_{im}-a_{im}}=\tilde{\lambda}_{m,k}a_{im}+b_{im}
\end{align}
where $(a)$ follows from the fact that $\log(x)\leq x-1<x,\;x>0$, (b) follows from the fact that whenever $\tilde{\lambda}_{m,k}>0$, $a_{ij}>a_{im}$ and $b_{ij}<b_{im}$, or  $a_{ij}<a_{im}$ and $b_{ij}>b_{im}$  holds. Hence, from \eqref{eq:verify}, we note that $\tilde{\lambda}_{m,k}$ is not a solution. Hence, the system of equations in \eqref{eq:lambda} is inconsistent. One can readily verify that \eqref{eq:om_lam} is inconsistent for any $|\mathcal{A}_T|> 3$. Hence, $|\mathcal{A}_T|\leq 2$, i.e., the number of layers in which the information is transmitted is at most two. Now, assuming that $\mathcal{A}_T=\{i,j\}$, the optimal $\lambda_{m,k}^*$ can be obtained from first equation in \eqref{eq:lambda} and $P_{i,k}^*$ can be obtained from  \eqref{eq:lam}. 
When $l_{i,k}^*\geq B_0/(P_{i,k}^*-U_k)$, the energy causality constraint gets violated in partition $i$. Hence, we consider the following case.

Next, we now assume that the energy stored in the battery is exhausted in $i$th partition, where $i$ is the second last partition in which the information is transmitted. 
The last partition uses only the power from the EH source. When only the EH power is used, one can prove that the optimal performance can be obtained by transmitting in only one partition. Hence, the information can be transmitted in partition $j\; (>i)$ only, in addition to partitions $1,\ldots,i$. In this case, we have, $\lambda_{1,k}=\ldots=\lambda_{i-1,k}=0, \lambda_{i,k}\geq 0$, $\lambda_{j,k}\geq 0$, $\beta_{j,k}=l_{j,k}>0$ and $e_{j,k}=0$. Substituting the values in \eqref{eq:lagrange_LTM} and differentiating with respect to $l_{j,k}$, we get, $\mu_k=q_j\log\left(1+h_jU_k\right)$.  Substituting $\mu_k$ in \eqref{eq:om_lam}, we note that the information can be transmitted only in partition $i$ among the initial $i$ partitions. Hence, $e_{i,k}=B_0$ and  $\nu_{e_{i,k}}=0$. Further, $\beta_{i,k}>0$ as charging and discharging the battery simultaneously is sub-optimal (See Lemma \ref{lemma:c-d}). From \eqref{eq:db_ltm} and \eqref{eq:de_ltm}, $\nu_{\beta_{i,k}}=\omega_{\beta_{i,k}}=0$ and $P_{i,k}=B_0/l_{i,k}+U_k$. From, \eqref{eq:db_ltm} and \eqref{eq:dl_ltm} and substituting  $\mu_k=q_j\log\left(1+h_jU_k\right)$, we obtain,  
\begin{align}\label{eq:mul}
\frac{q_ih_iP_{i,k}}{1+h_{i}P_{i,k}}  -q_i\log\left(1+h_iP_{i,k}\right)+q_j\log\left(1+h_jU_k\right)-\frac{q_ih_iU_k}{1+h_{i}P_{i,k}}=0
\end{align}
We can now solve for $l_{i,k}$ from \eqref{eq:mul}. Let $\mathcal{L}$ be the set of $l_{i,k}$'s that satisfy \eqref{eq:mul}. Then, the optimal, 
\begin{align}\label{eq:soll}
l_{i,k}^*=\argmax_{l_{i,k}\in \mathcal{L}} \left(  q_il_{i,k}\log(1+h_iP_{i,k})+q_j(\tau-l_{i,k})\log(1+h_jU_k)\right)
\end{align}
Now, one can easily obtain optimal $P_{i,k}^*$ and $l_{j,k}^*$ from $l_{i,k}^*$.

\subsection{Proof of Theorem \ref{lemma:LTM_SF_non_ideal} }
\paragraph{When $0<\phi_k^*<\tau$}
We first assume that $B_{\mathrm{max}}=\infty$. In this case, $\nu_{\phi_k}=\omega_{\phi_k}=0$, due to complementary slackness conditions. From \eqref{eq:dphi_ltm}, we have, $\mu_k=\mathcal{F}_c(V_{a_k}^*)\lambda_{i,k}$, where $i$ is the layer in which the battery energy is exhausted. If energy is allocated to any layer $i$ from the battery, i.e., $e_{i,k}>0$, then, $\nu_{e_{i,k}}=0$ and as charging and discharging the battery simultaneously is sub-optimal, we must have $\beta_{i,k}=l_{i,k}>0$ implying that $\nu_{\beta_{i,k}}=\nu_{l_{i,k}}=0$. Substituting $\omega_{\beta_{i,k}}$ from \eqref{eq:db_ltm}, $\lambda_{i,k}$ from \eqref{eq:de_ltm} and $\mu_k$ from \eqref{eq:dphi_ltm} in \eqref{eq:dl_ltm} and simplifying,
\begin{align}\label{eq:Pcphi}
h_i\left( x_{i,k}+\mathcal{F}_c(V_{a_k}^*)\right)\mathcal{F}_d'\left(x_{i,k}\right)-  \left(1+h_{i}P_{i,k} \right)\log\left(1+h_iP_{i,k}\right)=0
\end{align}
where $x_{i,k}=e_{i,k}/l_{i,k}$ and $P_{i,k}=U_k+\mathcal{F}_d(x_{i,k})-P_C$. 
%\eqref{eq:Pcphi} has a unique solution for any given layer $i$ considered one at a time. 
Clearly, \eqref{eq:Pcphi} is inconsistent if $\mathcal{A}_T=\{j_1,j_2,\ldots,j_n\},\; j_1\neq j_2\neq\ldots\neq j_n, 2\leq n\leq N$,  where $\mathcal{A}_T$ is the set of layers in which the information is transmitted. Hence, we conclude that in the optimal case, $|\mathcal{A}_T|=1$. Let $\mathcal{X}_i$ be the set of solutions to \eqref{eq:Pcphi} when $\mathcal{A}_T=i$.  Noting that $e_{i,k}=B_0+(\tau-l_{i,k})\mathcal{F}_c(V_{a_k}^*)$, we have, $\phi_k^*=l_{i,k}^*=\left({B_0+\tau \mathcal{F}_c(V_{a_k}^*)}\right)/\left({\tilde{x}_{i,k}+\mathcal{F}_c(V_{a_k}^*)}\right)$, where $\tilde{x}_{i,k}=\min{(\mathcal{X}_i)}$. 
Then, we choose  $i$ that maximizes $q_il_{i,k}\log(1+h_i(U_k-P_C+\mathcal{F}_d(\tilde{x}_{i,k})))$.  When $B_{\mathrm{max}}\neq \infty$, we have, $e_{i,k}=\min\left(B_0+(\tau-l_{i,k})\mathcal{F}_c(V_{a_k}^*),B_{\mathrm{max}}\right)$. 
Whenever $\phi_k^*\geq \tau$ in the above computation, it violates the frame length constraint and $\phi_k^*=\tau$ in the optimal case.  

\paragraph{When $\phi_k^*=\tau$}
 We have, $\nu_{\phi_k}=0$, $\omega_{\phi_k}\geq0$. We first assume that the energy is exhausted in  the last partition in which the transmission takes place. As in Appendix B,
\begin{align}\label{eq:phi1}
\omega_{\phi_k}=q_i\log\left(1+h_iP_{i,k}\right)-\frac{h_iq_i\left( g_{i,k}(\lambda_{k})+\mathcal{F}_c(V_{a_k}^*)\right)\mathcal{F}_d'\left(g_{i,k}(\lambda_{m,k})\right)}{\left(1+h_{i}P_{i,k} \right)}, \; \forall \; i\in \mathcal{A}_T
\end{align}
where $e_{i,k}/l_{i,k}=g_{i,k}(\lambda_{m,k})$ based on \eqref{eq:db_ltm} and $P_{i,k}=\mathcal{F}_d(g_{i,k}(\lambda_{m,k}))+U_k-P_C$, where $m$ is the last layer in which the battery is exhausted. It can be seen that when it is optimal to transmit in layers $i$ and $j$, \eqref{eq:phi1} is a system of equations with two variables and we can solve for unique $\omega_{\phi_k}$ and $\lambda_{m,k}$. From $\lambda_{m,k}$, we can find $e_{n,k}/l_{n,k}$ and $P_{n,k}$ for $n=i,j$.  

We now assume that the battery is exhausted in $i$th partition, where $i$ is the second last partition in which the information is transmitted. 
As in Appendix B, we have, 
\begin{align}\label{eq:lamPC}
\frac{h_iq_i\left( \frac{e_{i,k}}{l_{i,k}}+\mathcal{F}_c(V_{a_k}^*)\right)\mathcal{F}_d'\left(\frac{e_{i,k}}{l_{i,k}}\right)}{1+h_{i}P_{i,k}}-q_i\log\left(1+h_iP_{i,k}\right)+q_j\log\left(1+h_j(U_k-P_C)\right)=0
\end{align}
We now obtain the optimal $l_{i,k}^*$ from \eqref{eq:soll}, where the maximization is carried over the set of $l_{i,k}$'s that satisfy \eqref{eq:lamPC}. 
%Note that energy in the battery must be exhausted before the $m$th layer only if $e_{i,k}\geq B_0$ in the first case. 

\subsection{The Lagrangian of  $\mathrm{P_{LSC}}$ in \eqref{eq:p_lsc}  when $B_{\mathrm{max}}=\infty$ and Necessary Derivatives}
The Lagrangian of $\mathrm{P_{LSC}}$ in \eqref{eq:p_lsc}  when $B_{\mathrm{max}}=\infty$, 
\begin{align}\label{eq:lagrange_LSC}
L_{LSC}=&\sum_{k=1}^{K}\lambda_k \left( \phi_k \left(\sum_{i=1}^{N}s_i\mathcal{G}^{-1}(R_{i,k}/\phi_k)\prod_{l=1}^{i-1} \left(\mathcal{G}^{-1} \left(R_{l,k}/\phi_k\right)+1\right)+P_C-\frac{\beta_kU_k}{\phi_k}- \mathcal{F}_d\left(\frac{e_k}{\phi_k}\right) \right)\right)\nonumber\\
&+\sum_{k=1}^{K}\left(-\sum_{i=1}^{N}q_iR_{i,k}-\sum_{i=1}^{N}\mu_{i,k} R_{i,k} -\nu_{\phi_k} \phi_k +\omega_{\phi_k} (\phi_k-\tau)-\nu_{\beta_{k}}\beta_k+\omega_{\beta_{k}}(\beta_k-\phi_k)\right)\nonumber\\
&+\sum_{k=1}^{K}\psi_k\left(\sum_{j=1}^{k}\left(e_j-\phi_j\mathcal{F}_c(V_j)-(\tau-\phi_j) \mathcal{F}_c(V_{a_j}^*)\right)-B_0\right)-\sum_{k=1}^{K}\nu_{e_k}e_{k}
\end{align}
where $\lambda_k$'s and $\psi_k$'s are the non-negative Lagrange multipliers corresponding to \eqref{eq:p_lsc_c1} and \eqref{eq:p_lsc_c2}, respectively. $\mu_{i,k}$, $\nu_{\phi_k}$, $\omega_{\phi_k}$, $\nu_{\beta_k}$, $\omega_{\beta_k}$ and $\nu_{e_k}$ are non-negative Lagrange multipliers corresponding to inequalities $R_{i,k} \geq 0$, $\phi_k \geq 0$, $\phi_k-\tau\leq 0$, $\beta_k\geq 0$, $\beta_k-\phi_k\leq 0$ and $e_k\geq 0$, respectively, for each $k= 1,\ldots,K$ and $i=1,\ldots,N$.   

Differentiating $\mathrm{L_{LSC}}$ with  in \eqref{eq:lagrange_LSC} with respect to $R_{i,k}$, $e_k$, $\phi_k$ and $\beta_k$,
\begin{align}
&\frac{\partial L_{LSC}}{\partial R_{i,k}}=-q_i+\lambda_k \left(\sum_{l=i}^{N}\left(s_l-s_{l+1}\right)\exp(\frac{\sum_{j=1}^{l}R_{j,k}}{\phi_k})\right)-\mu_{i,k}=0 \label{eq:dR_lsc_sf}\\
&\frac{\partial L_{LSC}}{\partial e_k}=-\lambda_k\mathcal{F}_d'\left(\frac{e_k}{\phi_k}\right)-\nu_{e_k}+\sum_{j=k}^{K}\psi_k=0 \label{eq:de_lsc_sf}\\
&\frac{\partial L_{LSC}}{\partial \phi_k}=\lambda_k \left(\left(\sum_{i=1}^{N}\left(  \left(s_i-s_{i+1}\right)\exp(\frac{\sum_{j=1}^{i}R_{j,k}}{\phi_k})\right) \left(1-\frac{\sum_{j=1}^{i}R_{j,k}}{\phi_k}\right)\right)  -s_1+P_C\right)\nonumber\\
&\qquad\quad+\lambda_k\left(\frac{e_k}{\phi_k}\mathcal{F}_d'\left(\frac{e_k}{\phi_k}\right)-\mathcal{F}_d\left(\frac{e_k}{\phi_k}\right)\right)+\left(\frac{-\beta_kU_k\mathcal{F}_c'(V_k)}{\phi_k}-\mathcal{F}_c(V_k)+ \mathcal{F}_c(V_{a_k}^*)\right)\sum_{j=k}^{K}\psi_k \nonumber\\
&\qquad\quad-\nu_{\phi_k}+\omega_{\phi_k}-\omega_{\beta_{k}}=0 \label{eq:dphi_lsc_sf}\\
&\frac{\partial L_{LSC}}{\partial \beta_k}=-\lambda_kU_k+U_k\mathcal{F}_c'(V_k) \sum_{j=k}^{K}\psi_k-\nu_{\beta_{k}}+\omega_{\beta_{k}}=0 \label{eq:db_lsc_sf}
\end{align}

\subsection{Proof of Theorem \ref{thm:LSC_SF_ideal} }
First we prove that if $R_{a_{m},k}>0$ (equivalently $P_{a_{m},k}>0$ and $\mu_{a_m,k}=0$), then $R_{a_i,k}>0$ for all $i\geq m$  in the optimal solution.  This means that if a layer is allocated energy, then all the higher layers are also allocated energy. 
%From \eqref{eq:active_layers}
Assume that $R_{a_{m},k}>0$ which also implies $P_{a_{m},k}>0$ and $\mu_{a_m,k}=0$. We prove the result by contradiction.
Assume that $R_{a_i,k}=0,\; \forall i > m$.  Due to complementary slackness condition, we have, $\mu_{a_i,k}>0,\; \forall i > m$. From \eqref{eq:rate} and \eqref{eq:active_layers},
\begin{align}
\frac{-\mu_{a_{m+1},k}}{s_{a_m}-s_{a_{m+1}}} > \frac{\mu_{a_{m+1},k}-\mu_{a_{m+2},k}}{s_{a_{m+1}}-s_{a_{m+2}}}> \ldots >     \frac{\mu_{a_{A},k}}{s_{a_{A}}}
\end{align} 
Considering the first and the last terms, we can see that $\mu_{a_{A},k} < -\mu_{a_{m+1},k}s_{a_{A}}/(s_{a_m}-s_{a_{m+1}})<0$ which contradicts our assumption that $\mu_{a_{A},k}>0$.
Hence, we cannot have $R_{a_A,k}=0$. Similarly, we can consider other pairs and show that $R_{a_i,k}>0,\; \forall i > m$. 
% Hence, $\mu_{a_{A},k}$ cannot be strictly positive and we must have $\mu_{a_{A},k}=0$.  .   In the similar manner, we can prove that  $\mu_{a_i,k}=0,\; \forall i > m$  starting with any assumptions on $R_{a_i}$'s for $i>m$.
%Hence, we have proved that  if $R_{a_{m}}>0$ (equivalently $P_{a_{m}}>0$ and $\mu_{a_m,k}=0$), then $R_{a_i}>0$ and $\mu_{a_i,k}=0$ for all $i\geq m$  in the optimal solution. 
Now, we discuss the optimal power allocation.   
Assuming that the power is allocated starting from frame $a_m, m\geq 1$, then we have  $\mu_{a_i,k}=0$ for all $i\geq m$.  From \eqref{eq:rate}, we have,
\begin{align}\label{eq:la}
 \lambda_k \exp(\sum_{j=a_m}^{a_i}R_{j,k})=\frac{\tilde{p}_{a_i}}{s_{a_i}-s_{a_{i+1}}} \quad  i \geq m
\end{align}
For any frame $k \;(\geq m+1)$, from \eqref{eq:la}, we now evaluate $\lambda_k$ as
\begin{align}
\lambda_k=\frac{\tilde{p}_{a_k}}{(s_{a_k}-s_{a_{k+1}})\exp(\sum_{j=a_m}^{a_k}R_{j,k})}
\end{align}
Substituting the result in  \eqref{eq:la} for $k+1$, we obtain \eqref{eq:pmax}.   
%Note that we cannot compute $\exp(R_{a_{m}})$ in this manner. 
 Now, we can easily compute the transmit powers $\tilde{P}_{a_{k}},\;\forall k \geq m+1$.  So far, none of the terms consider the total power available, $P_k$. Even if $P_k$ is infinite, the above solution suggests that only $\tilde{P}_{a_{k}}$ units are allocated for layer $a_k$.  Hence, $\tilde{P}_{a_{k}}$'s can be treated as the maximum power allocated to any frame  $a_k$ and 
 %(? why cannot we allocate lower powers than Pmax?- Is it because it is suboptimal to allocate lower powers? If yes, why so?) 
we have, $	P^{\mathrm{max}}_{a_k} = \tilde{P}_{a_{k}} $. After, thus allocating the power to all the frames $k\geq m+1$, we allocate the remaining power to frame $a_m$. Since allocating the power to lower layers implies that it must be allocated to the higher layers as well, we allocate the power starting from the highest layer, $a_A$.

\subsection{Proof of Theorem \ref{thm:LSC_SF_non_ideal} }
\begin{itemize}[leftmargin=*]
	\item If  $\phi_k^*=0$, clearly, $R_{i,k}^*=P_{i,k}^*=0$,  for $i=1,\ldots, N$.  
%	This happens when the maximum combined instantaneous power from the battery and the EH source insufficient to operate the circuit, i.e., $\phi_k^*=0$ if $\mathcal{F}_d(B_0/\tau)+U_k\leq P_C$. 
	\item If $0<\phi_k^*<\tau$, due to complementary slackness condition, we have, $\nu_{\phi_k}=\omega_{\phi_k}=0$. Since the amount of energy drawn from the battery is non-zero, we have $e_k>0$ implying $\nu_{e_k}=0$. Since charging and discharging the battery simultaneously is sub-optimal (See Lemma \ref{lemma:c-d}), we have $\alpha_{b_k}^*=1$ implying that $\beta_k^*=\phi_k^*$ and $\omega_{\beta_k}\geq 0$. Expressing $\psi_k$ and $\omega_{\beta_k}$ in terms of $\lambda_k$ from \eqref{eq:de_lsc_sf} and \eqref{eq:db_lsc_sf}, respectively, and substituting them in \eqref{eq:dphi_lsc_sf}, we get,
	\begin{align}\label{eq:phi_lsc_sf}
	&\sum_{i=1}^{N}\left( \left(s_i-s_{i+1}\right)\exp(\frac{\sum_{j=1}^{i}R_{j,k}^*}{\phi_k^*})\right) \left(\frac{\sum_{j=1}^{i}R_{j,k}^*}{\phi_k^*}\right)\nonumber\\ &=\sum_{i=1}^{N}P_{i,k}+P_C - U_k-\mathcal{F}_d\left(\frac{e_k}{\phi_k}\right)+\mathcal{F}_d'\left(\frac{e_k}{\phi_k}\right)\left(\frac{e_k}{\phi_k}+\mathcal{F}_c(V_{a_k}^*)\right)
	\end{align}
	where we note  $\lambda_k>0$. 
	Since $\Sigma_{j=1}^{i}R_{j,k}^*/\phi_k^*= \mathcal{G}\left(({h_iP_{i,k}^*})/({1+h_i\sum_{j=i+1}^{N}P_{j,k}^*})\right)$ 
	is independent of $\phi_k^*$, the left-hand side of  \eqref{eq:phi_lsc_sf} is a function of $P_{i,k}^*$'s and $s_i$'s only. Hence, we use the following technique to obtain the solution: fix $B_0=P_C\tau$ and $U_k=0$. In this case, the optimal solution must have $0<\phi^*<\tau$ as the information cannot be transmitted in other cases. Let $\tilde{P}_{i,0}, \; i=1,\ldots,N$, be the optimal power allocation in this case. Now, from \eqref{eq:phi_lsc_sf}, we know that the optimal solution does not depend on $\phi_k^*$. Hence, for any value of $B_0$ and $U_k$, we can fix $\tilde{P}_{i,0}$'s as the optimal solution and vary $\phi_k^*$ such that the total power delivered is sufficient to run the circuitry and transmit the information, as follows. Let the transmission occur within $[\tau-\phi_k, \tau]$. 
	The total amount of energy stored in the battery during the non-transmission phase  is $e_k=B_0+(\tau-\phi_k)\mathcal{F}_c(V_{a_k}^*)$. Hence, the  power available for the transmission is $\mathcal{F}_d(e_k/\phi_k)+U_k$. To transmit the information in the optimal rate, we must have $\mathcal{F}_d(e_k/\phi_k)+U_k=\sum_{i=1}^{N}\tilde{P}_{i,k}+P_C$. The optimal $\phi_k^*$ is maximum $\phi_k$ that solves $\mathcal{F}_d(e_k/\phi_k)+U_k=\sum_{i=1}^{N}\tilde{P}_{i,k}+P_C$. When $B_{\mathrm{max}}\neq \infty$, we have, $e_k=\min\left(B_0+(\tau-\phi_k)\mathcal{F}_c(V_{a_k}^*),B_{\mathrm{max}}\right)$. 
	\item When $\phi_k^*=\tau$, the solution is similar to that in the ideal case with the available power equal to $U_k+\mathcal{F}_d\left(B_0/\tau\right)-P_C$.  Hence, the optimal solution is given by Theorem \eqref{thm:LSC_SF_ideal}. 
\end{itemize}

\subsection{Proof of Theorem \ref{thm:LSC_MF_non_ideal} }
We consider \eqref{eq:dR_lsc_sf}~--~\eqref{eq:db_lsc_sf} with $K=2$.  We now consider various cases on $\phi_k^*, \; k=1,2$. 
\subsubsection{$\phi_1^*=\phi_2^*=0$} In this case, $R_{i,k}^*=P_{i,k}^*=0$ for $i=1,\ldots, N$ and $k=1,\ldots, K$.
\subsubsection{$0<\phi_1^*,\phi_2^*<\tau$}
In this case, due to complementary slackness condition, we have $\nu_{\phi_k}=\omega_{\phi_k}=0$ for $k=1,2$. 
As in the single frame case, we have, $e_k>0, \alpha_{b_k}^*=1$, $\beta_k^*=\phi_k^*$  and  $\nu_{\beta_k}=0, \; k=1,2$.  From  \eqref{eq:de_lsc_sf}~--~\eqref{eq:db_lsc_sf}, we get equations with $\phi_k$ and $e_k$ as variable as in \eqref{eq:phi_lsc_sf} for $k=1,2$. The equations may have more than one solutions. However, the equations can be solved independently for $k=1,2$, along the lines of the proof in the single frame case, i.e., compute $\tilde{P}_{i,0}$ with $B_0=P_C\tau$ and assign $\tilde{P}_{i,1}=\tilde{P}_{i,2}=\tilde{P}_{i,0}$ for $i=1,\ldots,N$. We then select the optimal $\phi_k, \; k=1,2$ from Theorem \ref{thm:LSC_SF_non_ideal}.  
\subsubsection{$\phi_1=\tau,0<\phi_2^*<\tau$}
Clearly, for the second frame, $P^*_{i,2}=\tilde{P}_{i,0}$ for $i=1,\ldots,N$.  We now find the optimal $P_{i,1}$'s, $\phi_2^*$ and $e_k^*, \;k=1,2$ in the following by considering different cases.  
\paragraph{Case A}
We assume that the battery energy is used by both the frames, i.e., \eqref{eq:p_lsc_c2} is not satisfied with equality for $k=1$. 
Hence, $\beta_1=\tau$, $\beta_2=\phi_2$. Due to complementary slackness conditions, we have, $\omega_{\beta_k}\geq 0, \; k=1,2$, $\psi_1=\nu_{e_1}=\nu_{e_2}=0$. From \eqref{eq:dR_lsc_sf} and  \eqref{eq:de_lsc_sf}, we have, 
\begin{align}\label{eq:lam2}
&\psi_2=\lambda_1\mathcal{F}_d'\left(\frac{e_1}{\tau}\right)=\lambda_2\mathcal{F}_d'\left(\frac{e_2}{\phi_2}\right), \;\; q_N=\lambda_1s_K\exp(\frac{\sum_{j=1}^{N}R_{j,1}}{\phi_1})=\lambda_2s_K\exp(\frac{\sum_{j=1}^{N}R_{j,2}}{\phi_2})
\end{align}
Assuming $d_2=e_2/\phi_2$, from  \eqref{eq:lam2}, we have, 
\begin{align}\label{eq:d2e1}
{\mathcal{F}_d'\left(d_2\right)}{\exp(({\sum_{j=1}^{N}R_{j,1}})/{\tau})}={\mathcal{F}_d'\left({e_1}/{\tau}\right)}{\exp(({\sum_{j=1}^{N}R_{j,2}})/{\phi_2})} 
\end{align}
 In the second frame,  the total power required during the transmission is $\sum_{i=1}^{N}\tilde{P}_{i,0}+P_C$. Hence, we must have,  $\mathcal{F}_d(d_2)+U_2=\sum_{i=1}^{N}\tilde{P}_{i,0}+P_C$. Substituting the value of $d_2$ from this equation in \eqref{eq:d2e1}, we can solve for the unique $e_1^*$ based on Theorem \ref{thm:LSC_SF_ideal} within $N$ iterations subject to the battery capacity constraint.  
The amount of energy stored in the battery at the end of the first frame is $B_1=B_0-e_1$. Now, the optimal $\phi_2^*$ is the maximum $\phi_2$ that solves the following equation.  
\begin{align}\label{eq:phi2}
\mathcal{F}_d\left(\frac{\min\left(B_0-e_1+(\tau-\phi_2)\mathcal{F}_c(V_{a_2}^*),B_{\mathrm{max}}\right)}{\phi_2}\right)-P_C-\sum_{i=1}^{N}\tilde{P}_{i,0}=0
\end{align} 
\paragraph{Case B}
In the Case A, if $e_1<0$,  it is not optimal to allocate battery energy in the first frame. However,  it may be optimal to charge the battery in the first frame and transfer energy to the second frame, i.e., $e_1=0$ and $0<\beta_1\leq \tau$. Hence, due to complementary slackness condition, $\nu_{e_1}\geq 0$  and $\nu_{\beta_1}=\omega_{\beta_1}=0$.
 %The situation in the second frame does not change. 
  From \eqref{eq:de_lsc_sf} and \eqref{eq:db_lsc_sf}, we have,
\begin{align}\label{eq:lam1_beta1}
\lambda_1=\mathcal{F}_c'(V_1)\psi_2,\;\; \psi_2=\lambda_2\mathcal{F}_d'\left(\frac{e_2}{\phi_2}\right)
\end{align}
Recall that $V_k=(1-\beta_k/\phi_k)U_k$. 
From \eqref{eq:lam2} and \eqref{eq:lam1_beta1}, we have,
\begin{align}\label{eq:case3}
{\exp(\frac{\sum_{j=1}^{N}R_{j,2}}{\phi_2})}={\mathcal{F}_d'\left(d_2\right)}{\mathcal{F}_c'(V_1)\exp(\frac{\sum_{j=1}^{N}R_{j,1}}{\phi_1})}
\end{align}
We can obtain the unique $\beta_1$ by solving \eqref{eq:case3}  based on Theorem \ref{thm:LSC_SF_ideal} within $N$ iterations subject to battery capacity constraint.   
In this case, $\beta_1U_1/\tau$ \si{\watt} is used for transmission in the first frame. At the beginning of the second frame, we have $B_0+\mathcal{F}_c(V_1)\tau$ \si{\joule} in the battery. 
\subsubsection{$\phi_1^*=\phi_2^*=\tau$}
 In this case, the harvested energy is transferred from the first frame to the second frame. Hence, $e_1=0$, $0<\beta_1<\tau$, $e_2>0$ and $\beta_2=0$. Hence, $\nu_{e_1}$ may not be zero, $\nu_{\beta_1}=\omega_{\beta_1}=\nu_{e_1}=0$.  Hence, the solution can be obtained from \eqref{eq:case3} by substituting $d_2=(B_0+\mathcal{F}_c(V_1)\tau)/\tau$ and solving for the unique $\beta_1$. 
 
\bibliographystyle{ieeetran}
\bibliography{IEEEabrv,twireless}

\end{document}